\documentclass[lettersize,journal]{IEEEtran}
\usepackage{amsmath,amsfonts}
\usepackage{amsthm, amssymb}

\usepackage{algorithmic}
\usepackage{algorithm}
\usepackage{array}
\usepackage[caption=false,font=normalsize,labelfont=sf,textfont=sf]{subfig}
\usepackage{textcomp}
\usepackage{stfloats}
\usepackage{amssymb}
\usepackage{color}
\newtheorem{proposition}{Proposition}
\usepackage{hyperref}

\usepackage{url}
\usepackage{verbatim}
\usepackage{graphicx}
\usepackage{arydshln} % 用于虚线
\newcommand{\defeq}{\triangleq}

\usepackage{cite}
\theoremstyle{plain}
\newtheorem{theorem}{Theorem}
\newtheorem{lemma}{Lemma}
\newtheorem{remark}{Remark}

\hypersetup{colorlinks=true,
	linkcolor=blue,
	citecolor=blue,      
	urlcolor=black,
}

\begin{document}
	
	\title{DOA Estimation via Continuous Aperture Arrays: MUSIC and CRLB}
	
	\author{Haonan Si,~\IEEEmembership{Graduate Student Member, IEEE}, Zhaolin Wang, ~\IEEEmembership{Member, IEEE},\\	Xiansheng Guo, ~\IEEEmembership{Senior Member, IEEE}, Jin Zhang, ~\IEEEmembership{Member, IEEE}, Yuanwei Liu, ~\IEEEmembership{Fellow, IEEE}
	\thanks{Haonan Si and Xiansheng Guo are with the Department of Electronic Engineering, University of Electronic Science and Technology of China, Chengdu, 611731 China. (e-mail: sihaonan@std.uestc.edu.cn,  xsguo@uestc.edu.cn).}
	\thanks{Zhaolin Wang and Jin Zhang are with the School of Electronic Engineering and Computer Science, Queen Mary University of London, London E1 4NS, U.K. (e-mail: zhaolin.wang@qmul.ac.uk, jin.zhang@qmul.ac.uk).}
	\thanks{Yuanwei Liu is with the Department of Electrical and Electronic Engineering, The University of Hong Kong, Hong Kong. (e-mail: yuanwei@hku.hk).}
	}

	
	
	\maketitle

	\begin{abstract}
		Direction-of-arrival (DOA) estimation using continuous aperture array (CAPA) is studied. Compared to the conventional spatially discrete array (SPDA), CAPA significantly enhances the spatial degrees-of-freedoms (DoFs) for DOA estimation, but its infinite-dimensional continuous signals render the conventional estimation algorithm non-applicable. To address this challenge, a new multiple signal classification (MUSIC) algorithm is proposed for CAPAs. In particular, an equivalent continuous-discrete transformation is proposed to facilitate the eigendecomposition of continuous operators. Subsequently, the MUSIC spectrum is accurately approximated using the Gauss-Legendre quadrature, effectively reducing the computational complexity. Furthermore, the Cramér-Rao lower bounds (CRLBs) for DOA estimation using CAPAs are analyzed for both cases with and without priori knowledge of snapshot signals. It is theoretically proved that CAPAs significantly improve the DOA estimation accuracy compared to traditional SPDAs. Numerical results further validate this insight and demonstrate the effectiveness of the proposed MUSIC algorithm for CAPA. The proposed method achieves near-optimal estimation performance while maintaining a low computational complexity.
	\end{abstract}
	
	
	\section{Introduction}
	
	Estimating the direction-of-arrival (DOA) based on array signal processing stands out as a fundamental role in a broad range of applications including radar, acoustics, and seismology systems \cite{10934790,10858124,10945501,9810792}, which refers to the process of estimating the orientations of passive, non-cooperative targets via transmitting probing signals and processing the reflected their reflections.	Despite decades of research, it remains a persistent and challenging problem \cite{10144718,10643599}.
	
	The introduction of multi-input multi-output (MIMO) antenna arrays has significantly revolutionized sensing systems by enhancing the available uniform degrees of freedom (DoFs) \cite{9241013,7400949}. However, the spatial DoFs are still fundamentally limited by physical antenna deployment. A natural strategy to mitigate this limitation is to densely integrate more antenna elements into the constrained physical area of the transceiver. This motivation has driven the evolution from conventional MIMO to Massive MIMO and further to Gigantic MIMO architectures in next-generation wireless networks \cite{6824752,10298067}. Driven by these objectives, a variety of novel antenna array architectures have been proposed in recent years, including holographic MIMO systems \cite{10163760,9848831} and metasurface-based antennas \cite{10709869}. It is worth noting that the ultimate architecture of these arrays is envisioned as an (approximately) continuous electromagnetic (EM) aperture, which is defined as the \textit{continuous aperture array (CAPA)} \cite{liu2024capa}.
	
	\subsection{Prior Works}
	
		CAPA has emerged as a novel continuous‐aperture array architecture expected to transform 6G wireless communications by enabling fully analog beamforming across the entire bandwidth and delivering precise radiation‐pattern control \cite{liu2024capa,10938678}. Unlike conventional spatially discrete array (SPDA) systems—whose spatial degrees of freedom (DoFs) are inherently capped by the finite number of antenna elements—CAPA implements a continuous electromagnetic aperture densely populated with a vast number of  low-cost, miniaturized antennas. By realizing a continuous (or quasi-continuous) aperture, it fully harnesses every spatial DoF across its surface, thereby unlocking unprecedented channel capacity and sensing precision in a remarkably compact form factor \cite{9374451,10540217}. Based on rigorous analytical modeling, the authors in \cite{9139337} demonstrate that the spatial DoF in CAPA systems can exceed one, leading to a substantial increase in spatial capacity density. Driven by this promising advantage, extensive research efforts have been devoted to the modeling, performance analysis, and optimization of CAPA-enabled communication and sensing systems.
		
		Specifically, the authors of \cite{9906802} proposed a CAPA communication model grounded in wave propagation theory, introducing a wavenumber-division multiplexing (WDM) scheme for efficiently representing spatially continuous channels. The authors further analyzed the interference characteristics and implementation trade-offs under practical constraints. In addition, the authors of \cite{10612761} focused on uplink transmission in multi-user CAPA-assisted systems, where spectral efficiency was maximized through the optimization of continuous current density functions tailored to each user. Furthermore, the authors of \cite{10910020} investigated beamforming optimization in CAPA-based multi-user scenarios, deriving closed-form solutions for the optimal continuous source patterns and enhancing communication performance while reducing overall computational complexity.
	
		On the other hand, researchers have also investigated the fundamental performance limits of CAPA-based sensing systems. For instance, the authors in \cite{d2022cramer} analyzed these limits by deriving the Cramér-Rao lower bound (CRLB) for estimating the 3D Cartesian coordinates of target sources. Their theoretical analysis shows that CAPA can achieve centimeter-level positioning accuracy in the mmWave and sub-THz frequency bands. Furthermore, the authors of \cite{jiang2024cram} extend the theoretical framework to near-field sensing by deriving the CRLB for passive target positioning, leveraging the continuous transmit and receive array responses characteristic of CAPAs. To enhance positioning performance, the authors formulate an optimization problem that minimizes the CRLB through the design of continuous source currents, yielding significantly higher positioning accuracy compared to SPDA systems. Furthermore, the authors of \cite{10556596} investigated a more general sensing scenario by comprehensively evaluating the expected CRLB and Ziv-Zakai bound for joint target positioning and attitude estimation across diverse system configurations.

	\subsection{Motivation and Contributions}

	As aforementioned, recent research has demonstrated that CAPA systems can significantly enhance the available DoFs in sensing applications \cite{liu2024capa}, offering a promising avenue for improving DOA estimation performance. Several studies \cite{d2022cramer,jiang2024cram,10556596} have shown that CAPA systems can substantially boost estimation accuracy across various sensing tasks. Nonetheless, few algorithms have been specifically designed to fully leverage the unique characteristics of CAPA systems. For example, \cite{jiang2024cram} proposed a maximum likelihood estimation (MLE) method for CAPA-based near-field positioning, which suffers from degraded accuracy when multiple targets are present, thus limiting its practical applicability. Similarly, \cite{d2022cramer} discretized the CAPA aperture and applied MLE to evaluate estimation performance, but this approach fails to fully exploit the continuous nature and potential of CAPA systems. Hence, it is of significant meaning to design a novel high-accuracy DOA estimation method for CAPA systems.

	In conventional SPDA systems, multiple signal classification (MUSIC) is a widely adopted high-resolution DOA estimation algorithm that relies on the orthogonality between signal and noise sub-spaces \cite{6111312,9384289}. However, a fundamental challenge arises in CAPA systems, where the received signals are modeled as spatially continuous source current distributions over the aperture. Extracting target states from such infinite-dimensional measurements is inherently difficult. In particular, the MUSIC algorithm requires subspace decomposition of the data matrix, which becomes intractable in the infinite-dimensional setting of CAPA, rendering it unsuitable for CAPA-based sensing scenarios. These challenges motivate our investigation into the development of a MUSIC algorithm tailored for CAPA systems, along with an analysis of their fundamental performance limits.

	To the best of our knowledge, this is the first attempt to design a MUSIC algorithm for CAPA systems to fully exploit the spatial DoFs. Firstly, an equivalent continuous-discrete transformation is proposed to facilitate the eigendecomposition of continuous operators.  Then, the MUSIC spectrum is accurately approximated using Gauss-Legendre quadrature, effectively reducing computational complexity. Subsequently, the DOA of targets can be estimated by searching over the MUSIC spectrum. This further gives rise to an interesting and fundamental question: \textit{what performance limits for DOA estimation can CAPA achieve?} To answer this question, the CRLBs for DOA estimation using CAPAs are analyzed in both cases with or without the prior knowledge snapshot signals. It is theoretically proved that CAPAs significantly improve the DOA estimation accuracy compared to traditional SPDAs. The main contributions of this paper are summarized as follows:
	\begin{itemize}
		\item We formulate the problem of DOA estimation via CAPA within the framework of electromagnetic theory, where a receiving base station equipped with a CAPA is utilized to passively estimate the orientations of multiple targets.
	
		\item A new MUSIC algorithm is proposed for CAPA to estimate the DOAs of multiple targets. In particular, an equivalence continuous-discrete transformation is proposed to handle the intractable eigendecomposition of infinite-dimensional operators. Then, the MUSIC spectrum is approximated using Gauss-Legendre quadrature, contributing to high approximation accuracy while reducing the computational complexity.
	
		\item In practical scenarios, the snapshot signals received by the system may be either known or unknown. To comprehensively evaluate the fundamental performance of DOA estimation via CAPA, we derive the CRLBs for both cases. It is theoretically proved that CAPA significantly enhances DOA estimation accuracy compared to conventional SPDA systems.
	
		\item Extensive numerical experiments are conducted to validate the effectiveness and superiority of the proposed methods. The results show that the DOA estimation using CAPA substantially outperforms conventional SPDA systems. Moreover, the proposed MUSIC algorithm for CAPA achieves near-optimal estimation performance, with its mean squared error (MSE) closely approaching the CRLB. 
	\end{itemize}

	\subsection{Organization and Notations}
	The remainder of this paper is organized as follows: Section \ref{Sec:smpf} presents the CAPA-based DOA estimation model and problem formulation. Section \ref{sec:capade} elaborates the design of CAPA-MUSIC algorithm. The CRLB of CAPA-based DOA estimation is derived in Section \ref{sec:crlb}, encompassing both cases with known or unknown snapshot signals. Numerical results are presented in Section \ref{Sec::VNR} to demonstrate the superiority of CAPA-based DOA estimation and effectiveness of the proposed CAPA-MUSIC algorithm. Finally, Section \ref{sec:cons} concludes this paper.
	
	\textit{Notations:} Scalars, column vectors, matrices, and sets are denoted by $a$, $\mathbf{a}$, $\mathbf{A}$, and $\mathcal{A}$, respectively. $(\cdot)^\mathrm{T}$, $(\cdot)^\ast$, $(\cdot)^\mathrm{H}$, and $(\cdot)^{-1}$ denote the matrix transpose, conjugate, Hermitian, and  inverse operations, respectively. The real and imaginary parts of a complex variable $(\cdot)$ are denoted by $\Re{(\cdot)}$ and $\Im{(\cdot)}$, respectively. $\mathbb{R}^{M\times N}$ and $\mathbb{C}^{M\times N}$ mean the $M\times N$ real and complex matrix spaces, respectively. The notation $[(\cdot)]_{m,n}$ refers to the $(m,n)$-th entry of a matrix or vector. $||(\cdot)||_2$ denotes the Euclidean norm of vector $(\cdot)$. $|\mathcal{S}|$ represents the Lebesgue measure  of an Euclidean subspace $\mathcal{S}$.  The imaginary unit is represented by $\mathrm{j}=\sqrt{-1}$.

	\section{System Model and Problem Formulation}\label{Sec:smpf}

	The system model of this work is illustrated in Fig. \ref{fig:arch}. In the considered system, the base station is equipped with a CAPA to receive the signals. Without loss of generality, the CAPA is placed on the XOY plane, with the center set at the origin of the coordinate system. Its two sides are parallel to the $x$- and $y$-axes, with length along these axes being $L_x$ and $L_y$, respectively. Suppose there are a total of $M$ targets, for which the DOA are to be estimated. The location of the $m$-th target is denoted by $\mathbf{q}_m=[q_m^{(x)},q_m^{(y)},q_m^{(z)}]^{\mathrm{T}}\in\mathbb{R}^3$, where $m\in\mathcal{M}$ and $\mathcal{M}=\{1,2,...,M\}$. These targets are assumed to be located in the radiating far-field region or Fraunhofer radiation region \cite{orfanidis2002electromagnetic} of the CAPA; that is, the minimum distance from the CAPA to the target $z_0$ satisfying that $z_0> 2D^2/\lambda$, where $D$ and $\lambda$ denote the aperture size and the wavelength, respectively.
	
	\begin{figure}[t]
		\centering
		\includegraphics[width=0.7\linewidth]{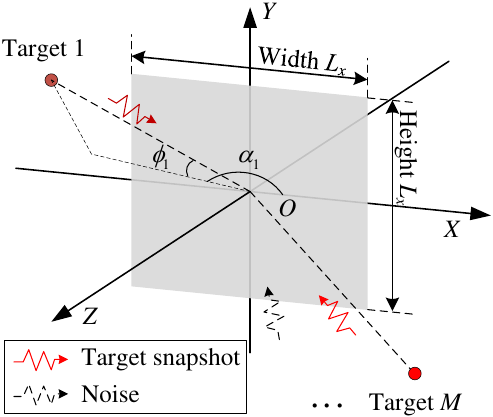}
		\caption{An illustration of the CAPA-based DOA estimation system.}
		\label{fig:arch}
	\end{figure}

	Under the configuration shown in Fig. \ref{fig:arch}, the size of CAPA is denoted by $L_x\times L_y$. A point on the CAPA is expressed as $\mathbf{r}=[r_x,r_y,r_z]\in \mathcal{S}$ where $\mathcal{S}$ represents the coordinate regions of CAPA specified by
	\begin{align}
			\label{Eq:aper}
				\mathcal{S}=&\left\{
				\mathbf{r} \defeq [r_x,r_y,r_z]^{\mathrm{T}}: -\frac{L_x}{2}\leq r_x\leq\ \frac{L_x}{2}, \right. \nonumber\\ &\quad\left.  -\frac{L_y}{2}\leq r_y\leq \frac{L_y}{2}, r_z=0
				\right\}.
	\end{align}

	\subsection{Channel Modeling}

	To characterize the mapping from a prescribed current source 	distribution \( J(\varpi, \mathbf{q}) \) to the resulting radiated electromagnetic field \( E(\varpi, \mathbf{q}) \), where \( \varpi \) denotes the angular frequency, we invoke the inhomogeneous Helmholtz equation in the temporal-frequency domain \cite{9765526}:
	\begin{equation}
		\label{eq:helmholtz_freq}
		\nabla^2 E(\varpi, \mathbf{q}) + \left( \frac{\varpi}{c} \right)^2 E(\varpi, \mathbf{z}) = \mathrm{j} \varpi \mu_0 J(\varpi, \mathbf{q}),
	\end{equation}
	where \( \mu_0 \) is the vacuum permeability and \( c \) is the speed of light in free space. This equation models the wave propagation in a linear, homogeneous, isotropic, and unbounded medium, driven by the source term \( J(\varpi, \mathbf{q}) \).

%	The system described by \eqref{eq:helmholtz_freq} is linear and space-time invariant, allowing its spectral behavior to be captured via a space-time Fourier transform. Specifically, by transforming both sides into the joint temporal-frequency and spatial-frequency (wavenumber) domain, the output spectrum is given by:
%	\begin{equation}
%		E(\varpi, \mathbf{k}) = H(\varpi, \mathbf{k}) J(\varpi, \mathbf{k}),
%	\end{equation}
%	where \( \mathbf{k} = [k_x, k_y, k_z]^\mathrm{T} \) denotes the spatial frequency vector (also referred to as the wavenumber), and \( H(\varpi, \mathbf{k}) \) is the wavenumber-frequency response of the system, defined as:
%	\begin{equation}
%		H(\varpi, \mathbf{k}) = \frac{\mathrm{j} \varpi \mu_0}{\|\mathbf{k}\|^2 - \left( \frac{\varpi}{c} \right)^2}.
%	\end{equation}

	In the context of narrowband and single-carrier systems, the frequency \( \varpi \) can be considered fixed. Under this assumption, the system becomes time-invariant, and the analysis can be restricted to its spatial behavior. The system is then fully characterized by its spatial impulse response \( h(\mathbf{q}) \), which satisfies the following inhomogeneous Helmholtz equation:
	\begin{equation}
		\label{eq:inhomogeneous_helmholtz}
		\nabla^2 h(\mathbf{q}) + k^2 h(\mathbf{q}) = \mathrm{j} k \eta_0 \delta(\mathbf{q}),
	\end{equation}
	where \( k = 2\pi / \lambda \) is the wavenumber and \( \delta(\mathbf{q}) \) denotes the three-dimensional Dirac delta function, respectively. The solution to \eqref{eq:inhomogeneous_helmholtz} is well-known and is given by:
		\begin{equation}
			h(\mathbf{q}) = -\mathrm{j} k \eta_0 G(\mathbf{q}),
		\end{equation}
		where \( G(\mathbf{q}) \) is the scalar Green's function of the Helmholtz operator in free space, expressed as:
		\begin{equation}
			G(\mathbf{q}) = \frac{e^{-\mathrm{j} k \|\mathbf{q}\|_2}}{4\pi \|\mathbf{q}\|_2}.
		\end{equation}
		Physically, this represents an outgoing spherical wave originating from a point source located at the origin, with radial propagation and amplitude decay proportional to \( 1 / \|\mathbf{q}\|_2 \).

		Then, by adding all the waves from the sources in the free space, the channel impulse response at $\mathbf{r}\in\mathcal{S}$ can be formulated as
		\begin{align}
			E(\mathbf{r})&=\int_{\mathbb{R}^3}J(\mathbf{q})h(\mathbf{r}-\mathbf{q})d\mathbf{q}\nonumber\\
			&\stackrel{(a)}{=}\sum_{m=1}^{M}J(\mathbf{q}_m)h(\mathbf{z}_m),
		\end{align}
		where equation $(a)$ is because there are a total of $M$ sources in the considered system, $\mathbf{z}_m=\mathbf{r}-\mathbf{q}_m$ denotes the receive vector for the $m$-th target with $\mathbf{r}\in\mathcal{S}$, and the receive antenna response can be formulated as follows:
	\begin{equation}
		\label{eq:grzm}
		h(\mathbf{z}_m)=\frac{\mathrm{j}k\eta_0e^{-\mathrm{j}k||	\mathbf{z}_m||_2}}{4\pi||\mathbf{z}_m||_2}.
	\end{equation}	
	This result describes the spatial propagation model between the receiver and targets.
	
	\subsection{Planar Wave Approximation}
	To facilitate the following signal modeling and algorithm design, we expand $||\mathbf{z}_m||_2$ by using the planar wave approximation. $||\mathbf{z}_m||_2$ can be formulated as:
	\begin{align}
		\label{eq:zm2}
			||\mathbf{z}_m||_2&=||\mathbf{r}-\mathbf{q}_m||_2\nonumber\\
			&=||\mathbf{q}_m||_2\sqrt{1-2\frac{\left(\mathbf{r}\cdot\bar{\mathbf{q}}_m\right)}{||\mathbf{q}_m||_2}+\frac{||\mathbf{r}||_2^2}{||\mathbf{q}_m||_2^2}}\nonumber\\
			&\stackrel{(a)}{\approx} ||\mathbf{q}_m||_2-\left(\mathbf{r}\cdot\bar{\mathbf{q}}_m\right)+\frac{\left[1+\left(\mathbf{r}\cdot\bar{\mathbf{q}}_m \right)^2\right]||\mathbf{r}||_2^2}{2||\mathbf{q}_m||_2},
	\end{align}
	where $\bar{\mathbf{q}}_m=\mathbf{q}_m/||\mathbf{q}_m||_2$ denotes the direction vector of $\mathbf{q}_m$, $\mathbf{r}\cdot\bar{\mathbf{q}}_m$ represents the inner product operation between vectors, and inequality $(a)$ is due to the fact that the Taylor series expansion $\sqrt{1+x}\approx 1+x/2-x^2/8$  is valid for small values of $x$.	By defining the wave vector  $\mathbf{d}(\alpha_m,\phi_m)=[\cos\alpha_m\cos\phi_m,\sin\alpha_m\cos\phi_m,\sin\phi_m]^{\mathrm{T}}$, one has $\mathbf{r}\cdot\bar{\mathbf{q}}_m=\mathbf{r}^{\mathrm{T}}\mathbf{d}(\alpha_m,\phi_m)$. Then, substituting Eq. (\ref{eq:zm2}) into Eq. (\ref{eq:grzm}) and neglecting the 2-order term results in:
	\begin{equation}
		h(\mathbf{z}_m)\approx  \frac{\mathrm{j}k\eta_0e^{-\mathrm{j}k||	\mathbf{q}_m||_2}}{4\pi||\mathbf{q}_m||_2} e^{\mathrm{j}k\mathbf{r}^{\mathrm{T}} \mathbf{d}({\alpha_m}, {\phi_m})},
	\end{equation}
	where $\alpha_m$ and $\phi_m$ denote the azimuth and elevation angles, respectively. 
	
	\begin{figure}[t]
		\centering
		\includegraphics[width=0.5\linewidth]{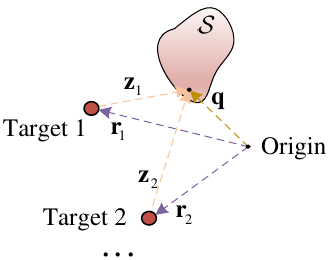}
		\caption{An explanation of geometry of targets.}
		\label{fig:geo}
	\end{figure}

	Then, due to the presence of noise, the excited electric filed at time $t$ at the point $\mathbf{r}\in \mathcal{S}$ can be written as follows:
	\begin{align}
		\label{eq:receive_approx}
			E(\mathbf{r},t)=& \sum_{m=1}^{M}J(\mathbf{q}_m,t)h(\mathbf{z}_m)+n(\mathbf{r},t)\nonumber\\
			=&\sum_{m=1}^{M} s_m(t) e^{\mathrm{j} k \mathbf{r}^{\mathrm{T}} \mathbf{d}({\alpha_m}, {\phi_m})}+n(\mathbf{r},t),
	\end{align}
	where $s_m(t)$ denotes the snapshot signal of the $m$-th target, formulated as
	\begin{equation}
		\label{eq:sm}
		s_m(t)=J(\mathbf{q}_m,t)\frac{\mathrm{j}k\eta_0e^{-\mathrm{j}k||	\mathbf{q}_m||_2}}{4\pi||\mathbf{q}_m||_2}.
	\end{equation}	
	The noise term $n(\mathbf{r},t)$ is assumed to  be a spatially uncorrelated zero-mean complex Gaussian
	process as described in \cite{jensen2008capacity}. The correlation function of this process is formulated as
	\begin{align}
			&\mathbb{E}\left\{n(\mathbf{r}_1,t),n^\ast(\mathbf{r}_2,t)\right\}=\sigma^2\delta(\mathbf{r}_1-\mathbf{r}_2),\nonumber\\
			 &\forall \mathbf{r}_1, \mathbf{r}_2\in \mathcal{S},
	\end{align}
	where $\sigma^2$ denotes the noise covariance and $\delta(\cdot)$ is the Dirac function, indicating that the noise is uncorrelated across the CAPA aperture. Notably, our work can be extended to the situation with spatially colored noise, further details can be found in \cite{10612761}. Defining that $\mathbf{s}(t)=[s_1(t),s_2(t),...,s_M(t)]^{\mathrm{T}}$, the DOA parameters $\boldsymbol{\alpha}=\left[\alpha_1,\alpha_2,...,\alpha_M\right]^\mathrm{T}$ and $\boldsymbol{\phi}=\left[\phi_1,\phi_2,...,\phi_M\right]^\mathrm{T}$.  Subsequently, Eq. (\ref{eq:receive_approx}) can be written in a compact form:
	 \begin{equation}
	 	E(\mathbf{r},t)={\mathbf{a}}^{\mathrm{T}}(\mathbf{r};\boldsymbol{\alpha},\boldsymbol{\phi}) \mathbf{s}(t)+n(\mathbf{r},t),
	 \end{equation} 
	 where $\mathbf{a}(\mathbf{r};\boldsymbol{\alpha},\boldsymbol{\phi})$ is formulated as
	 \begin{align}
	 		&\mathbf{a}(\mathbf{r};\boldsymbol{\alpha},\boldsymbol{\phi})=\left[a(\mathbf{r};\alpha_1,\phi_1),a(\mathbf{r};\alpha_2,\phi_2),...,a(\mathbf{r};\alpha_M,\phi_M)\right]^\mathrm{T},\\
	 		&a(\mathbf{r};\alpha_m,\phi_m)=e^{\mathrm{j} k \mathbf{r}^{\mathrm{T}} \mathbf{d}({\alpha_m}, {\phi_m})}, m\in\mathcal{M}.
	 \end{align}

	To facilitate the subsequent algorithm design and theoretical analysis, the continuous aperture of the CAPA receiver is equivalently transformed into an infinite number of infinitesimal non-overlapping receiving units, i.e., $\mathcal{S} = \bigcup_{n=1}^{N} \mathcal{S}_n$, where $N \to \infty$. The voltage observed at the output of the $n$-th element is obtained by integrating over the region as
	\begin{align}
		\label{eq:xndefi}
			x_n(t)=&\int_{\mathcal{S}_n}E(\mathbf{r},t)d\mathbf{r}\nonumber\\
			\approx&|\mathcal{S}_n|\left[\mathbf{a}^{\mathrm{T}}(\mathbf{r}_n;\boldsymbol{\alpha},\boldsymbol{\phi}) \mathbf{s}(t)+n(\mathbf{r}_n,t)\right]\nonumber\\
			=&|\mathcal{S}_n|\mathbf{a}^{\mathrm{T}}(\mathbf{r}_n;\boldsymbol{\alpha},\boldsymbol{\phi}) \mathbf{s}(t)+\nu(\mathbf{r}_n,t),
	\end{align}
	where $n\in\mathcal{N}$,  $\mathcal{N}=\{1,2,...,N\}$, $\mathbf{r}_n\in\mathcal{S}_n$ is an arbitrary point within the $n$-th region, $|\mathcal{S}_n|$ denotes the area of region $\mathcal{S}_n$, and $\nu(\mathbf{r}_n,t)$ is independent zero-mean Gaussian random variables with variance $\sigma_{\nu}^2=|\mathcal{S}_n|\sigma^2$. Without loss of generality, it is assumed that the area of each region is uniform and equals to $Z$, i.e., $|\mathcal{S}_n|=Z$ for $\forall n\in\mathcal{N}$. By concatenating the received signals, steering vectors, and the noises across $n\in\mathcal{N}$, we derive the following vectors or matrices with infinite dimensions (keep in mind that $N\rightarrow \infty$):
		\begin{align}
			&\mathbf{x}(t) = \left[x(\mathbf{r}_1,t),x(\mathbf{r}_2,t),...,x(\mathbf{r}_N,t)\right]^{\mathrm{T}}\in \mathbb{C}^{N},\\ 
			&\mathbf{A}(\boldsymbol{\alpha},\boldsymbol{\phi}) = \left[\mathbf{a}(\mathbf{r}_1;\boldsymbol{\alpha},\boldsymbol{\phi}),\mathbf{a}(\mathbf{r}_2;\boldsymbol{\alpha},\boldsymbol{\phi}),...,\mathbf{a}(\mathbf{r}_N;\boldsymbol{\alpha},\boldsymbol{\phi})\right]^{\mathrm{T}}\nonumber\\
			&\quad\quad\quad\quad\in \mathbb{C}^{M\times N},\\
			&\mathbf{n}(t) = \left[ n(\mathbf{r}_1,t),n(\mathbf{r}_2,t),...,n(\mathbf{r}_N,t)\right]^{\mathrm{T}}\in \mathbb{C}^{N}.
		\end{align}
	Then, the received signal can be rewritten in the following compact form:
	\begin{equation}
		\mathbf{x}(t)=\lim\limits_{Z\rightarrow 0}Z\left[\mathbf{A}(\boldsymbol{\alpha},\boldsymbol{\phi})\mathbf{s}(t)+\mathbf{n}(t)\right],
	\end{equation}
	where $Z=\lim\limits_{N\rightarrow +\infty}|\mathcal{S}_n|\rightarrow 0$, due to the fact that CAPA has theoretically infinite array units.
	
	In the considered CAPA-based DOA estimation scenario, the target transmits probing signals with source current function $J(\mathbf{q}_m)$ and the CAPA receives the signal $x(\mathbf{r})$ over region $\mathcal{S}$. Hence, the aim of this paper is to extract the targets' orientation information $\{\boldsymbol{\alpha}, \boldsymbol{\phi}\}$ from the received signal ${x}(\mathbf{r},t)$ over the entire CAPA. Since the azimuth and elevation angles are of interest, we mainly focus on the estimation accuracy.
	
	\begin{remark}\normalfont
		The goal of DOA estimation is to determine the angles $\{\boldsymbol{\alpha},\boldsymbol{\phi}\}$ from the measurements $\mathbf{x}(1), \mathbf{x}(2), \dots, \mathbf{x}(T)$ at $T$ time instances. Traditionally, the MUSIC algorithm is employed for high-resolution DOA estimation by exploiting the orthogonality of the noise subspace. However, in the context of the CAPA system under consideration,  the measurements are of infinite dimension, i.e., $N\rightarrow \infty$, making traditional methods no longer applicable. This motivates us to design a modified MUSIC algorithm tailored for CAPA systems.
	\end{remark}

	\section{MUSIC Algorithm Design for CAPAs}\label{sec:capade}
	
	In this section, we aim to design a MUSIC algorithm tailored for the CAPA systems, which is denoted by CAPA-MUSIC algorithm. The proposed CAPA-MUSIC method is a subspace-based DOA-estimation method for multiple targets with unknown signal snapshots. 
	
	\subsection{Subspace Decomposition for Spatially-Continuous Signal}
	Assuming that the sources are non-coherent, the noise $\mathbf{n}\sim \mathcal{CN}(0,\sigma^2\mathbf{I}_N)$ is an additive white Gaussian noise, $\mathbf{I}_N$  is an identity matrix with $N$ dimension, and $\mathbf{n}$ is independent with the signals, the covariance matrix of $\mathbf{x}(t)$ can be formulated as	
	\begin{align}\label{Eq::RX1}
			\mathbf{R}_{\mathbf{x}}=&\lim\limits_{Z \rightarrow 0}\mathbb{E}\left[\mathbf{x}\mathbf{x}^{\mathrm{H}}\right]\nonumber\\
			=&\lim\limits_{Z \rightarrow 0}\mathbf{A}\mathbf{R}_{\mathbf{s}} \mathbf{A}^{\mathrm{H}}+\sigma^2\mathbf{I}_{N},
	\end{align}
	where $\mathbf{R}_{\mathbf{x}}\in \mathbb{R}^{N\times N}$  and $\mathbf{R}_{\mathbf{s}}\in \mathbb{R}^{M\times M}$ denote the covariance matrices of the received signal $\mathbf{x}$ and snapshot signal $\mathbf{s}$, respectively.
	Theoretically, $\mathbf{R}_{\mathbf{x}}$ can be reformulated by eigenvalue decomposition:
	\begin{equation}\label{Eq::RX2}
		{\mathbf{R}}_{\mathbf{x}}=\lim\limits_{Z \rightarrow 0}{\mathbf{U}}{\mathbf{\Lambda}}{\mathbf{U}}^{\mathrm{H}},
	\end{equation}
	where ${\mathbf{U}}\in\mathbb{R}^{N\times N}$ is a unitary matrix satisfying ${\mathbf{U}}{\mathbf{U}}^{\mathrm{H}}={\mathbf{I}}_N$ and $\mathbf{\Lambda}\in\mathbb{R}^{N\times N}$ is a diagonal matrix. Then, by combining Eqs. (\ref{Eq::RX1}) and (\ref{Eq::RX2}), the following equation holds:
	\begin{align}
		\label{eq:eqei}
			\mathbf{A}\mathbf{R}_{\mathbf{s}} \mathbf{A}^{\mathrm{H}}=&{\mathbf{U}}{\mathbf{\Lambda}}{\mathbf{U}}^{\mathrm{H}}-\sigma^2{\mathbf{U}}{\mathbf{U}}^{\mathrm{H}}\nonumber\\
			=&{\mathbf{U}}({\mathbf{\Lambda}}-\sigma^2{\mathbf{I}}_N){\mathbf{U}}^{\mathrm{H}}.
	\end{align}
	
	Notably, the ranks on both sides of Eq. (\ref{eq:eqei}) should be equal \cite{9384289}. The rank of left side of Eq. (\ref{eq:eqei}) is equal to the dimension of $\mathbf{R}_\mathbf{s}$, i.e., the number targets $M$, while the rank of $\mathbf{U}$ is $N$. This observation indicates that the entries of matrix ${\mathbf{\Lambda}}-\sigma^2{\mathbf{I}}_N$ should be zero except the first $M\times M$ entries. Define that ${\mathbf{U}}=
	\left[	
	\begin{array}{l:l}
		\mathbf{U}_{11} &{{\mathbf{U}}_{12}} \\ \hdashline
		{{\mathbf{U}}_{21}} & {\mathbf{U}}_{22}
	\end{array}
	\right]
	$, with $\mathbf{U}_{11}\in \mathbb{R}^{M\times M}$, $\mathbf{U}_{11}\in \mathbb{R}^{M\times M}$, $\mathbf{U}_{12}\in \mathbb{R}^{M\times (N-M)}$, $\mathbf{U}_{21}\in \mathbb{R}^{(N-M)\times M}$, and $\mathbf{U}_{22}\in  \mathbb{R}^{(N-M)\times (N-M)}$ being the sub-matrices of ${\mathbf{U}}$. Let $\mathbf{U}_{\mathcal{S}}=[{\mathbf{U}}_{11};\mathbf{U}_{21}]$ and $\mathbf{U}_{\mathcal{N}}=[{\mathbf{U}}_{12};\mathbf{U}_{22}]$ denote the eigenvectors of signal and noise sub-spaces, respectively. Based on the observation above, it can be concluded that 
	\begin{equation}
		\text{span}(\mathbf{A})=\text{span}(\mathbf{U}_{\mathcal{S}}),
	\end{equation}
	where $\text{span}(\cdot)$ denotes the space spanned by the column vectors of matrix $(\cdot)$. Due to the fact that signal sub-space and the noise sub-space are orthogonal, i.e., $\text{span}(\mathbf{U}_{\mathcal{S}})\perp \text{span}(\mathbf{U}_{\mathcal{N}})$, it can be further derived that:	
	\begin{equation}\label{eq:bsth}
		\text{span}(\mathbf{A})\perp\text{span}(\mathbf{U}_{\mathcal{N}})\Rightarrow{\mathbf{U}}_{\mathcal{N}}^{\mathrm{H}}\mathbf{A}=\mathbf{0},
	\end{equation}
	which is the basic principle of MUSIC algorithm.
	
	In practical applications, we define the estimated covariance matrix for $\mathbf{X}$:
	\begin{align}
			\hat{\mathbf{R}}_{\mathbf{x}}=&\lim\limits_{Z \rightarrow 0}\frac{1}{T}\sum_{t=1}^{T}\mathbf{x}(t)\mathbf{x}^{\mathrm{H}}(t)\nonumber\\=&\lim\limits_{Z \rightarrow 0}\frac{1}{T}\mathbf{X}\mathbf{X}^{\mathrm{H}},
	\end{align}
	where $\mathbf{X}=[\mathbf{x}(1),\mathbf{x}(2),...,\mathbf{x}(T)]\in\mathbb{R}^{N\times T}$. Notably, it should be guaranteed that $T\geq M$ because at least $M$ snapshots are required to derive a rank $M$ approximation of $\mathbf{R}_\mathbf{x}$ \cite{9384289}.

	In traditional MUSIC algorithms \cite{10144718}, $\mathbf{U}_{\mathcal{N}}$ can be obtained by executing eigenvalue decomposition on $\hat{\mathbf{R}}_{\mathbf{x}}$. However, in the considered CAPA-based DOA estimation system, $N$ is an infinite constant and  $\hat{\mathbf{R}}_{\mathbf{x}}$ has infinite eigenvalues, rendering the calculation of $\mathbf{U}_{\mathcal{N}}$ challenging. To solve this challenge, we propose a novel eigenvalue decomposition method for infinite matrices in the sequel.
	
	\subsection{Eigendecomposition for Spatially-Continuous Signal}
	
	Based on the definition of eigenvalues and eigenvectors,  we have $\hat{\mathbf{R}}_{\mathbf{x}}\mathbf{u}=\kappa\mathbf{u}$, where $\mathbf{u}$ and $\kappa$ denote the eigenvector and the corresponding eigenvalue of matrix $\hat{\mathbf{R}}_{\mathbf{x}}$, respectively.  Multiplying on both side with $\mathbf{X}^{\mathrm{H}}$ results in:	
	\begin{equation}\label{eig:eq}
		\frac{1}{T}\lim\limits_{Z \rightarrow 0}\mathbf{X}^{\mathrm{H}}\mathbf{X}\mathbf{X}^{\mathrm{H}}{\mathbf{u}}=\lim\limits_{Z \rightarrow 0}\lambda\mathbf{X}^{\mathrm{H}}{\mathbf{u}},
	\end{equation}

	Defining $\mathbf{K}=\lim\limits_{Z \rightarrow 0}\mathbf{X}^{\mathrm{H}}\mathbf{X}/T\ \in \mathbb{R}^{T\times T}$ and $\mathbf{e}=\lim\limits_{Z \rightarrow 0}\mathbf{X}^{\mathrm{H}}{\mathbf{u}}\in \mathbb{R}^{T\times 1}$, Eq. (\ref{eig:eq}) can be rewritten as
	\begin{equation}\label{eigEQ}
		\mathbf{K}\mathbf{e}=\lambda\mathbf{e}.
	\end{equation}

	\begin{remark}\normalfont
		\label{rem:rem1}
		Notably, matrices $\hat{\mathbf{R}}_{\mathbf{x}}$ and $\mathbf{K}$ share $T$ common eigenvalues. Thus, the eigenvalue computation for $\hat{\mathbf{R}}_{\mathbf{x}}$ equivalently carried out by calculating those of the finite matrix $\mathbf{K}$. It is worth pointing that the operation in Eq. (\ref{eig:eq}) inevitably involves some information loss, as the number of eigenvalues is reduced from infinity to $T$. However, since the rank of $\hat{\mathbf{R}}_{\mathbf{x}}$ does not exceed $T$, the remaining eigenvectors can be expressed as linear combinations of the first $T$-th eigenvectors. Thus, this transformation is both effective and crucial to the proposed eigenvalue decomposition method, which will be further demonstrated by numerical results in Section \ref{Sec::VNR}.
	\end{remark}	
	
	Specifically, the elements of $\mathbf{K}$ can be calculated as
	\begin{align}\label{element_K}
			[\mathbf{K}]_{ij}=&\lim\limits_{Z \rightarrow 0}
			\frac{1}{T} \sum_{n=1}^{N} x^\ast(\mathbf{r}_n,i)x(\mathbf{r}_n,j)\nonumber\\
			=&\frac{1}{T}\lim_{Z\rightarrow 0}Z^2 \sum_{n=1}^{N} E^\ast(\mathbf{r}_n,i)E(\mathbf{r}_n,j)\nonumber\\
			=&\frac{1}{T} \int_{-\frac{L_y}{2}}^{\frac{L_y}{2}} \int_{-\frac{L_x}{2}}^{\frac{L_x}{2}} E^\ast(\mathbf{r},i)E(\mathbf{r},j)dr_xdr_y.
	\end{align}

	The integral of a function $f(x,y)$ across a surface can be computed by using the second order Gauss-Legendre quadrature, which takes the following form \cite{ralston2001first}:
	\begin{equation}\label{Gauss_l}
		\int_{-1}^1\int_{-1}^1 f(x,y) dxdy \approx \sum_{k_x=1}^K\sum_{k_x=1}^K \omega_{k_x}\omega_{k_y} f\left(\vartheta_{k_x},\vartheta_{k_y}\right),
	\end{equation}
	where $K$ is the dimension of Gauss-Legendre quadrature representing the number of sample points. A larger value of $K$ results in higher approximation accuracy. $\vartheta_{k_x}$ and $\vartheta_{k_y}$  are the roots of Gauss-Legendre polynomial. $\omega_{k_x}$ and $\omega_{k_y}$ denote the weight of Gauss-Legendre polynomial, which can be formulated as follows:
	\begin{align}
			&\omega_{k_x}=\frac{2}{(1-\vartheta_{k_x}^2)[P_K'(\vartheta_{k_x})]^2},k_x\in\mathcal{K},\label{eq:glq0}\\
			&\omega_{k_y}=\frac{2}{(1-\vartheta_{k_y}^2)[P_K'(\vartheta_{k_y})]^2}, k_y\in\mathcal{K},
		\label{eq:glq}
	\end{align}
	where $\mathcal{K}=\left\{1,2,...,K\right\}$, $P_K'(\vartheta_{k_x})$ and $P_K'(\vartheta_{k_y})$ denotes the the first-order differentials of the $K$-th Gauss-Legendre polynomial evaluated at $\vartheta_{k_x}$ and $\vartheta_{k_y}$, respectively. This formula ensures that the quadrature rule achieves high accuracy for polynomials. Based on Gauss-Legendre quadrature, Eq. (\ref{element_K})  can be rewritten as follows:
	\begin{align}\label{simK}
			[\mathbf{K}]_{ij}
			=&\frac{1}{T} \int_{-\frac{L_y}{2}}^{\frac{L_y}{2}} \int_{-\frac{L_x}{2}}^{\frac{L_x}{2}} E^\ast(\mathbf{r},i)E(\mathbf{r},j)dr_xdr_y\nonumber\\
			\approx&\frac{L_xL_y}{4T} \sum_{k_x=1}^{K}\sum_{k_y=1}^{K} \omega_{k_x}\omega_{k_y} E^\ast\left(\mathbf{r}_{k_x,k_y},i\right) \nonumber\\
			&\quad\times E\left(\mathbf{r}_{k_x,k_y},j\right),
	\end{align}
	where $\mathbf{r}_{k_x,k_y}=\left[\vartheta_{k_x}L_x/2,\vartheta_{k_y}L_y/2,0\right]^{\mathrm{T}}$ denotes the sample points, $\omega_{k_x}$ and $\omega_{k_y}$ represent the weights  of the $K$-th Gauss-Legendre polynomial.
	
%	Let  $\boldsymbol{\Omega}=\text{diag}\left\{\omega_1,\omega_2,...,\omega_K\right\}$ denote the weight matrix. Then, we have $\mathbf{K}= \mathbf{\bar{X}}^{\mathrm{H}}\boldsymbol{\Omega}\mathbf{\bar{X}}/N$, where  $\mathbf{\bar{X}}=[\mathbf{\bar{x}}(1),\mathbf{\bar{x}}(2),...,\mathbf{\bar{x}}(N)]$ and $\mathbf{\bar{x}}(i)=\left[x(\mathbf{r}_1,i),x(\mathbf{r}_2,i),...x(\mathbf{r}_K,i)\right]^{\mathrm{T}}$, respectively.
	
	Subsequently, we can calculate the eigenvectors of matrix $\mathbf{K}$, denoted by $\mathbf{e}_1,\mathbf{e}_2,...,\mathbf{e}_T$, which are sorted in descending order according to their corresponding eigenvalues. Similarly, let $\mathbf{u}_{i}$ denote the $i$-th eigenvector of $\mathbf{X}\mathbf{X}^{\mathrm{H}}$. As analyzed in \textbf{Remark \ref{rem:rem1}}, we only consider the $T$ eigenvectors with the highest eigenvalues. Theoretically, $\mathbf{u}_{i}\in\mathbb{R}^N$ and each dimension corresponds to an array element in CAPA, i.e., $\mathbf{u}_{i} = \left[u_{i}(\mathbf{r}_1),u_{i}(\mathbf{r}_2),...,u_i(\mathbf{r}_N)\right]\in \mathbb{C}^{N}$, where $i\in\mathcal{T}$ and $\mathcal{T}=\{1,2,...,T\}$. Similar to Eqs. (\ref{element_K})-(\ref{simK}), the $t$-th element of the $i$-th eigenvector $[\mathbf{e}_i]_t$ can be approximated by
	\begin{align}\label{simE}
		[\mathbf{e}_i]_t=&\lim\limits_{Z \rightarrow 0}\sum_{n=1}^{N}x^\ast(\mathbf{r}_n,i)u_i(\mathbf{r}_n)\nonumber\\
		=&\int_{\mathcal{S}}E^\ast(\mathbf{r},k)u_i(\mathbf{r})d\mathbf{r}\nonumber\\
		=& \int_{-\frac{L_y}{2}}^{\frac{L_y}{2}} \int_{-\frac{L_x}{2}}^{\frac{L_x}{2}} E^\ast(\mathbf{r},t)u_i(\mathbf{r})dr_xdr_y\nonumber\\
		\approx&\frac{L_xL_y}{4} \sum_{k_x=1}^{K}\sum_{k_x=1}^{K} \omega_{k_x}\omega_{k_y} u_i\left(k_x,k_y\right) \nonumber\\
		&\times E^\ast\left(\mathbf{r}_{k_x,k_y},t\right)= \mathbf{E}^{\mathrm{H}}(t)\boldsymbol{\Omega} \bar{\mathbf{u}}_j,
	\end{align}
	where $\boldsymbol{\Omega}=\frac{L_xL_y}{4}\boldsymbol{\omega}\otimes\boldsymbol{\omega}$, $\otimes$ represents the Kronecker product, and 
	\begin{align}\label{elesup}
			&\boldsymbol{\omega}=\text{diag}\left\{\omega_1,\omega_2,...,\omega_K\right\}\in\mathbb{R}^{K^2\times K^2},\\
			\label{elesup1}&\mathbf{E}(t)=\left[ E(\mathbf{r}_{1,1},t),E(\mathbf{r}_{1,2},t),...,E(\mathbf{r}_{K,K},t)\right]^{\mathrm{T}}\in\mathbb{R}^{K^2},\\
			&\mathbf{\bar{u}}_i=\left[u_i(1,1),u_i(1,2),...,u_i(K,K)\right]^{\mathrm{T}}\in\mathbb{R}^{K^2}.
	\end{align}
	By defining $\bar{\mathbf{E}}=[\mathbf{E}(1),\mathbf{E}(2),...,\mathbf{E}(T)]$, we have $\mathbf{e}_i\approx\bar{\mathbf{E}}^{\mathrm{H}}\boldsymbol{\Omega}\bar{\mathbf{u}}_i$, such that the $i$-th eigenvector of $\hat{\mathbf{R}}_{\mathbf{x}}$ can be approximated by:
	\begin{equation}
		\label{eq:ukbar}
		\mathbf{\bar{u}}_i\approx\left(\mathbf{\bar{E}}^{\mathrm{H}}\boldsymbol{\Omega}\right)^\dag\mathbf{e}_i, i\in\mathcal{T},
	\end{equation}
	where $(\mathbf{E}^{\mathrm{H}}\boldsymbol{\Omega})^{\dag}$ denotes the Moore-Penrose pseudo inverse. Here, we have finished the eigenvalue decomposition for infinite matrix $\hat{\mathbf{R}}_\mathbf{x}$, based on equivalent transformation and pseudo inverse operations.
	
%	\begin{remark}
%		Adhere, we utilize the Legendre polynomial and Moore-Penrose pseudo-inverse to derive the eigenvectors for an infinite matrix $\hat{\mathbf{R}}_{\mathbf{x},\infty}$. 
%	\end{remark}
	
	\subsection{CAPA-MUSIC spectrum}
	
	Then, the eigenvectors of the noise subspace are formulated as
	\begin{equation}
		\label{eq:unbar}
		\bar{\mathbf{U}}_{\mathcal{N}}=[\mathbf{\bar{u}}_{M+1},\mathbf{\bar{u}}_{M+2},...,\mathbf{\bar{u}}_T].
	\end{equation}
	Finally, based on Eq. (\ref{eq:bsth}), the DOA parameters $\{\alpha,\phi\}$ is estimated by identifying the peaks in the MUSIC spectrum for CAPA:
	\begin{align}
		\label{music_spe}
			\{\hat{\alpha},\hat{\phi}\}=&{}^M\arg\max_{\alpha,\phi}P_{\text{MUSIC}}(\alpha,\phi)\nonumber\\
			=&{}^M\arg\max_{\alpha,\phi}\lim_{Z\rightarrow 0}\frac{1}{\boldsymbol{a}^{\mathrm{H}}(\alpha,\phi)\mathbf{U}_{\mathcal{N}}\mathbf{U}^{\mathrm{H}}_{\mathcal{N}}\boldsymbol{a}(\alpha,\phi)}\nonumber\\
			\stackrel{(a)}{\approx}&{}^M\arg\max_{\alpha,\phi}\frac{1}{\boldsymbol{\bar{a}}^{\mathrm{H}}(\alpha,\phi)\boldsymbol{\Omega}\mathbf{\bar{U}}_{\mathcal{N}}\mathbf{\bar{U}}_{\mathcal{N}}^{\mathrm{H}}\boldsymbol{\Omega}\boldsymbol{\bar{a}}(\alpha,\phi)},
	\end{align}
	where ${}^M\arg\max f(\cdot)$ denotes the set of $M$ values at which the function $f(\cdot)$ reaches its local maxima, and
	\begin{align}
			&\boldsymbol{{a}}(\alpha,\phi)\nonumber\\
			= &\left[ {a}(\mathbf{r}_{1};\alpha,\phi),{a}(\mathbf{r}_{2};\alpha,\phi),...,{a}(\mathbf{r}_{N};\alpha,\phi)\right]^{\mathrm{T}}\in\mathbb{R}^{N},\\
			&\boldsymbol{\bar{a}}(\alpha,\phi)\nonumber\\
			=&\left[ {a}(\mathbf{r}_{1,1};\alpha,\phi),{a}(\mathbf{r}_{1,2};\alpha,\phi),...,{a}(\mathbf{r}_{K,K};\alpha,\phi)\right]^{\mathrm{T}}\in\mathbb{R}^{K^2}.
	\end{align}
	 Approximation $(a)$ is based on Gauss-Legendre quadrature, similar to the operations in Eqs. (\ref{simK}) and (\ref{simE}).  This transfers the computation involving infinite-dimensional vectors into the one in a finite-dimensional space.

	Algorithm \ref{alg:CAPA-MUSIC} summarized the CAPA-MUSIC algorithm.

	\begin{algorithm}[t]
		\caption{CAPA-MUSIC Algorithm}
		\begin{algorithmic}[1]\label{alg:CAPA-MUSIC}
			\REQUIRE CAPA aperture parameters (area $L_x\times L_y$ and region $\mathcal{S}$ defined by its boundaries);
			Number of snapshots $T$ and the received signal $x(\mathbf{r},t)$;
			Dimension of Legendre polynomial $K$.

			\ENSURE Estimated target azimuth and elevation angles $\{\hat{\alpha}_m,\hat{\phi}_m\}_{m=1}^M$.
			\STATE Compute the roots $\vartheta_k$ and weights $\omega_k$ for $K$-th Legendre polynomial using Eqs. \eqref{eq:glq0} and (\ref{eq:glq}).
			\STATE Assemble the elements of $\mathbf{K}$ by using Eq. (\ref{simK}).
			\STATE Compute eigenvectors $\mathbf{e}_1,\mathbf{e}_2,...,\mathbf{e}_T$ using eigen-decomposition.
			\STATE  Compute $\boldsymbol{\Omega}$ and ${\mathbf{E}}(t)$ using Eq. (\ref{elesup}) and \eqref{elesup1}.
			\STATE Compute $\bar{\mathbf{U}}_\mathcal{N}$  using Eqs. (\ref{eq:ukbar}) and (\ref{eq:unbar}).
			\STATE Compute the MUSIC spectrum using Eq. (\ref{music_spe}).
			\STATE Search on the MUSIC spectrum to find the optimal $\hat{\alpha},\hat{\phi}$.
			\RETURN Estimated target azimuth and elevation angles $\{\hat{\alpha}_m,\hat{\phi}_m\}_{m=1}^M$.
		\end{algorithmic}
	\end{algorithm}
	
	\subsection{Complexity Analysis}
	The computational complexity of the proposed CAPA-MUSIC algorithm is analyzed as follows. Firstly, the calculation of $\vartheta_{k}$ and $\omega_k$ in Step 1 of Algorithm \ref{alg:CAPA-MUSIC} has complexity $\mathcal{O}(K^2)$. The eigen-decomposition in Step 3 takes up complexity of $\mathcal{O}(T^3)$. Deriving the eigenvectors of noise sub-space in steps 4-5 has complexity $\mathcal{O}(K^4T)$. For each scan grid, the calculation of $P_{\text{MUSIC}}$ in step 6 consumes complexity of $\mathcal{O}(K^4)$. Let $N_S$ denote the number of scan grids, the total computational  complexity of CAPA-MUSIC is $\mathcal{O}(T^3+K^4T+N_S K^4)$.
	
	Typically, a larger number of scan grids contribute to achieve high angular resolution, while resulting in heavier complexity. To alleviate this challenge, an alternative method is to perform a low-resolution coarse grid scan with a small $N_\theta$, and use a gradient-based optimizer to estimate the DOA with high resolution \cite{9384289}.

	\section{Cramér-Rao Lower Bound}\label{sec:crlb}
	In this section, we evaluate the fundamental DOA estimation performance of CAPA and quantify its performance gain relative to traditional SPDA-based methods. Specifically, we consider two scenarios corresponding to different real-world applications: one where the snapshot signals $\left\{\mathbf{s}(t)\right\}_{t=1}^{T}$ are known, and another where $\left\{\mathbf{s}(t)\right\}_{t=1}^{T}$ are unknown.
	
	Firstly, based on Eq. (\ref{eq:xndefi}), the probability density of $\mathbf{x}(1),\mathbf{x}(2),...,\mathbf{x}(T)$ is formulated as
	\begin{align}\label{proden}
			&p\left(
			\mathbf{x}(1),\mathbf{x}(2),...,\mathbf{x}(T)\bigg| \boldsymbol{\alpha},\boldsymbol{\phi},\left\{\mathbf{s}(t)\right\}_{t=1}^{T}
			\right)\nonumber\\
			=&\left|\pi\sigma_{\nu}^2{\mathbf{I}}_N\right|^{-T}\exp \left\{
			-\sigma_{\nu}^{-2} \sum_{t=1}^{T}\left\| \mathbf{x}(t)-Z\mathbf{A}(\boldsymbol{\alpha},\boldsymbol{\phi})\mathbf{s}(t)\right\|_2^2\right\}.
	\end{align}	
	To derive the CRLB for DOA estimation, consider the likelihood function of $\mathbf{x}(1),\mathbf{x}(2),...,\mathbf{x}(T)$ given $\{\boldsymbol{\alpha},\boldsymbol{\phi}\}$ and $\left\{\mathbf{s}(t)\right\}_{t=1}^T$, which is given by
	\begin{align}\label{eq:ll}
			&\mathcal{L}\left(\boldsymbol{\alpha},\boldsymbol{\phi},\left\{\mathbf{s}(t)\right\}_{t=1}^T\right)\nonumber\\=&\lim\limits_{Z \rightarrow 0} -\frac{1}{\sigma_{\nu}^2} \sum_{t=1}^{T}\sum_{n=1}^{N}\left[{x}(\mathbf{r}_n,t)- Z{\mathbf{a}}^{\mathrm{H}}(\mathbf{r}_n;\boldsymbol{\alpha},\boldsymbol{\phi}) \mathbf{s}(t)\right]^2	\nonumber\\
			&\quad\quad -T\log(\left|\pi\sigma_{\nu}^2{\mathbf{I}}_N\right|) \nonumber\\
			=&- \frac{1}{\sigma^2_{\nu}}\sum_{t=1}^{T}\int_{\mathcal{S}}
			\left[ {E}(\mathbf{r},t)-\mathbf{a}^{\mathrm{H}}(\mathbf{r};\boldsymbol{\alpha},\boldsymbol{\phi})\mathbf{s}(t) \right]^2\,d\mathbf{r}+C,
	\end{align}
	where $C$ is a constant irrelevant to $\boldsymbol{\theta}$ and $\left\{\mathbf{s}(t)\right\}_{t=1}^T$.
	
	\subsection{The Case With Known Target Snapshots}
	In this case, the parameters to be estimated is $\boldsymbol{\theta}=\{\boldsymbol{\alpha},\boldsymbol{\phi}\}$ and  we only focus on the likelihood function with regard to $\boldsymbol{\theta}$. To obtain the closed-loop form of CRLB,  we define  $f(\mathbf{r};\boldsymbol{\theta})=\left[ {x}(\mathbf{r},t)-\mathbf{a}^{\mathrm{H}}(\mathbf{r};\boldsymbol{\alpha},\boldsymbol{\phi})\mathbf{s}(t) \right]^2$ and introduce the following lemma:
	
	\begin{lemma}\normalfont
		\label{lemma:1}
		$f(\mathbf{r};\boldsymbol{\theta})$ satisfies the regularity condition, i.e., the integration and differentiation can be exchanged:
		\begin{equation}
			\label{regcon}
			\frac{\partial}{\partial [\boldsymbol{\theta}]_i} \int_{\mathcal{S}} f(\mathbf{r};\boldsymbol{\theta}) d \mathbf{r}=\int_{\mathcal{S}} \frac{\partial f(\mathbf{r};\boldsymbol{\theta})}{\partial [\boldsymbol{\theta}]_i} d \mathbf{r},
		\end{equation}		
		where $i\in\mathcal{Q}$ and $\mathcal{Q}=\left\{1,2,...,2M\right\}$.		
		The proof can be found in Appendix \ref{App:A}.
	\end{lemma}
	
	In this case, the parameters to be estimated is $\boldsymbol{\theta}$ and we aim to derive the following symmetric Fisher Information Matrix (FIM) $\mathbf{J}_{\boldsymbol{\theta}\boldsymbol{\theta}}\in \mathbb{R}^{2M\times 2M}$:
	\begin{align}\label{eq:FIM_element2}
			\left[\mathbf{J}_{\boldsymbol{\theta}\boldsymbol{\theta}}\right]_{i,j} &=
			\mathbb{E}\left\{
			\frac{\partial \mathcal{L}}{\partial[\boldsymbol{\theta}]_i}
			\frac{\partial \mathcal{L}}{\partial[\boldsymbol{\theta}]_j}
			\right\}\nonumber\\
			&\stackrel{(a)}{=} \frac{2}{\sigma^2_{\nu}}\sum_{t=1}^{T}\Re\left\{
			\int_{\mathcal{S}} \left(\frac{\partial \mathbf{a}^\mathrm{H}(\mathbf{r};\boldsymbol{\alpha},\boldsymbol{\phi})}{\partial[\boldsymbol{\theta}]_i}\mathbf{s}(t)\right)^*\right.\nonumber\\
			&\quad\quad\quad\quad\quad\quad\quad\left.\left(\frac{\partial \mathbf{a}^{\mathrm{H}}(\mathbf{r};\boldsymbol{\alpha},\boldsymbol{\phi})}{\partial[\boldsymbol{\theta}]_j}\mathbf{s}(t)\right)
			d\mathbf{r}
			\right\},
	\end{align}
	where equation $(a)$ is because of \textbf{Lemma \ref{lemma:1}} and $i,j\in\mathcal{Q}$.
	
	Then, we define $R_{ij}=\mathbb{E}\left[\sum_{t=1}^{N}s_i(t)s_j(t)\right]$ and divide the FIM matrix into $\mathbf{J}_{\alpha\alpha},\mathbf{J}_{\alpha\phi},\mathbf{J}_{\phi\phi}\in \mathbb{R}^{M\times M}$ as follows:
	\begin{equation}\label{eq:FIM_final2}
		\left[\mathbf{J}_{\boldsymbol{\theta}\boldsymbol{\theta}}\right] = \frac{2}{\sigma^2_{\nu}} 
		\left[\begin{array}{l:l}
			\mathbf{J}_{\alpha\alpha} & 
			\mathbf{J}_{\alpha\phi} \\ \hdashline
			\mathbf{J}_{\alpha\phi}^{\mathrm{T}} & 
			\mathbf{J}_{\phi\phi}\\
		\end{array}\right],
	\end{equation}
	where $\mathbf{J}_{\alpha\alpha}=\mathbf{J}_{\alpha\alpha}^\mathrm{T}$, $\mathbf{J}_{\phi\phi}=\mathbf{J}_{\phi\phi}^\mathrm{T}$, and $\mathbf{J}_{\alpha\phi}=\mathbf{J}_{\phi\alpha}^\mathrm{T}$. Specifically, the elements of these matrices are calculated by
	\begin{align}
		\label{FIM:elem}
			&\left[\mathbf{J}_{\alpha\alpha}\right]_{i,j}\nonumber\\=&R_{ij}\int_{\mathcal{S}}\Re\left\{\left(\frac{\partial a(\mathbf{r};\alpha_i,\phi_i)}{\partial\alpha_i}\right)^\ast
			\left(\frac{\partial a(\mathbf{r};\alpha_j,\phi_j)}{\partial\alpha_j}\right)\right\}d\mathbf{r},\\
			&\left[\mathbf{J}_{\alpha\phi}\right]_{i,j}\nonumber\\=&R_{ij}\int_{\mathcal{S}}\Re\left\{\left(\frac{\partial a(\mathbf{r};\alpha_i,\phi_i)}{\partial\alpha_i}\right)^\ast
			\left(\frac{\partial a(\mathbf{r};\alpha_j,\phi_j)}{\partial\phi_j}\right)\right\}d\mathbf{r},\\
			&\left[\mathbf{J}_{\phi\phi}\right]_{i,j}\nonumber\\=&R_{ij}\int_{\mathcal{S}}\Re\left\{\left(\frac{\partial a(\mathbf{r};\alpha_i,\phi_i)}{\partial\phi_i}\right)^\ast
			\left(\frac{\partial a(\mathbf{r};\alpha_j,\phi_j)}{\partial\phi_j}\right)\right\}d\mathbf{r}.
	\end{align}
	For $i\in\mathcal{M}$, the derivative of the steering vector with respect to $\alpha_i$ and $\phi_i$ can be respectively formulated as 
	\begin{align}
		\label{eq:q00}
			&\frac{\partial a(\mathbf{r};\alpha_i,\phi_i)}{\partial\alpha_i}
			= \mathrm{j} k \left(\mathbf{r}^{\mathrm{T}}\frac{\partial \mathbf{d}(\alpha_i,\phi_i)}{\partial\alpha_i}\right)
			e^{\mathrm{j} k\,\mathbf{r}^{\mathrm{T}}\mathbf{d}(\alpha_i,\phi_i)},\\
			&\frac{\partial a(\mathbf{r};\alpha_i,\phi_i)}{\partial\phi_i}
			= \mathrm{j} k \left(\mathbf{r}^{\mathrm{T}}\frac{\partial \mathbf{d}(\alpha_i,\phi_i)}{\partial\phi_i}\right)
			e^{\mathrm{j} k\,\mathbf{r}^{\mathrm{T}}\mathbf{d}(\alpha_i,\phi_i)}.
	\end{align}	
	Since the CAPA is located on the plane $r_z=0$, any point $\mathbf{r}\in\mathcal{S}$ can be represented as
	\begin{equation}
		\label{eq:q0qz}
		\mathbf{r}=[r_x,r_y,0]^{\mathrm{T}}.
	\end{equation}	
	By using Eq. (\ref{eq:q0qz}) and recalling that $\mathbf{d}(\alpha_m,\phi_m)=[\cos\alpha_m\cos\phi_m,\sin\alpha_m\cos\phi_m,\sin\phi_m]^{\mathrm{T}}$, for $\forall i\in\mathcal{M}$, we have
	\begin{align}
		\label{eq:qtqyq}
			\mathbf{r}^{\mathrm{T}}\frac{\partial \mathbf{d}(\alpha_i,\phi_i)}{\partial \alpha_i}
			&= r_x\cos\alpha_i\cos\phi_i, \\
			\mathbf{r}^{\mathrm{T}}\frac{\partial \mathbf{d}(\alpha_i,\phi_i)}{\partial \phi_i}
			&\label{eq::555}= -r_x\sin\alpha_i\sin\phi_i+r_y\cos\phi_i.
	\end{align}	
	Then, substituting Eqs. (\ref{eq:q00})-(\ref{eq::555}) into Eq. (\ref{eq:FIM_final2}) gives rise to Eq. (\ref{eq:eq35f}), as shown in the bottom of this page.	
	\begin{figure*}[b]
		\hrulefill
		\begin{align}
			\label{eq:eq35f}
				&\left[\mathbf{J}_{\alpha\alpha}\right]_{ij}=k^2R_{ij}\int_{-\frac{L_y}{2}}^{\frac{L_y}{2}}\int_{-\frac{L_x}{2}}^{\frac{L_x}{2}}\left[r_x^2\cos\alpha_i\cos\phi_i\cos\alpha_j\cos\phi_j\cos\left(k\left(\Delta d_yr_x+\Delta d_zr_y\right)\right)\right]dr_xdr_y,\\
			\label{eq:eq36f}	&\left[\mathbf{J}_{\alpha\phi}\right]_{ij}=k^2R_{ij}\int_{-\frac{L_y}{2}}^{\frac{L_y}{2}}\int_{-\frac{L_x}{2}}^{\frac{L_x}{2}}\left[( r_x\cos\alpha_i\cos\phi_i)(-r_x\sin\alpha_j\sin\phi_j+r_y\cos\phi_j)\cos\left(k\left(\Delta d_yr_x+\Delta d_zr_y\right)\right)\right]dr_xdr_y,\\		
			\label{eq:eq37f}	&\left[\mathbf{J}_{\phi\phi}\right]_{ij}=k^2R_{ij}\int_{-\frac{L_y}{2}}^{\frac{L_y}{2}}\int_{-\frac{L_x}{2}}^{\frac{L_x}{2}}\left[(-r_x\sin\alpha_i\sin\phi_i+r_y\cos\phi_i)(-r_x\sin\alpha_j\sin\phi_j+r_y\cos\phi_j)\cos\left(k\left(\Delta d_yr_x+\Delta d_zr_y\right)\right)\right]dr_xdr_y.
		\end{align}
	\end{figure*}	
	The CRLB for any unbiased estimator of $\hat{\boldsymbol{\theta}}=[\hat{\boldsymbol{\alpha}},\hat{\boldsymbol{\phi}}]^{\mathrm{T}}$ is given by
	\begin{equation}\label{eq:CRLB_final}
		\mathrm{Cov}(\hat{\boldsymbol{\theta}}) \ge (\mathbf{J}_{\boldsymbol{\theta}\boldsymbol{\theta}})^{-1},
	\end{equation}
	where $\hat{\boldsymbol{\alpha}}=\left[\hat{\alpha}_1,\hat{\alpha}_2,...,\hat{\alpha}_M\right]$ and $\hat{\boldsymbol{\phi}}=\left[\hat{\phi}_1,\hat{\phi}_2,...,\hat{\phi}_M\right]$. For each target $m\in\mathcal{M}$, $\hat{\alpha}_m$ and $\hat{\phi}_m$ denote the estimated azimuth and elevation angles, satisfying that:
	\begin{align}
		\mathrm{Var}(\hat{\alpha}_m) &\ge \left[(\mathbf{J}_{\boldsymbol{\theta}\boldsymbol{\theta}})^{-1}\right]_{m,m},\\[1mm]
		\mathrm{Var}(\hat{\phi}_m)  &\ge \left[(\mathbf{J}_{\boldsymbol{\theta}\boldsymbol{\theta}})^{-1}\right]_{M+m,M+m}.
	\end{align}
	%	where $\left[\mathbf{J}(\alpha,\phi)^{-1}\right]_{1,1}\in R^{M\times M}$ denotes the matrix consisting of the first $M$ rows and $M$ columns of $\mathbf{J}(\alpha,\phi)^{-1}$, $\left[\mathbf{J}(\alpha,\phi)^{-1}\right]_{2,2}\in R^{M\times M}$ denotes the matrix consisting of the last  $M$ rows and $M$ columns of $\mathbf{J}(\alpha,\phi)^{-1}$, respectively.

	\begin{theorem}\normalfont
		Under the situation that target snapshots $\left\{\mathbf{s}(t)\right\}_{t=1}^T$ are known, the CRLB for CAPA-based DOA estimation of the $m$-th target is
		\begin{align}\label{eq::crlb}
				&\mathrm{CRLB}_k(\alpha_m) = \left[(\mathbf{J}_{\boldsymbol{\theta}\boldsymbol{\theta}})^{-1}\right]_{m,m}, m\in\mathcal{M},\\
				&\mathrm{CRLB}_k(\phi_m) = \left[(\mathbf{J}_{\boldsymbol{\theta}\boldsymbol{\theta}})^{-1}\right]_{M+m,M+m}, m\in\mathcal{M},
		\end{align}
		where the subscript $k$ is utilized to highlight that the above results refer to the case with known target snapshots and the definition of $\mathbf{J}_{\boldsymbol{\theta}\boldsymbol{\theta}}$ are given in Eqs.  (\ref{eq:FIM_final2}) and (\ref{eq:eq35f}).
	\end{theorem} 
		
	This completes the derivation of the CRLB for DOA estimation using a continuous aperture CAPA system.
	
	\subsection{The Case With Unknown Signal Snapshots}

	In this situation, the parameters to be estimated consist of the DOA $\boldsymbol{\theta}$, real parts of the snapshots $\mathbf{s}(t)$, and image parts of the snapshots $\mathbf{s}(t)$ for $t\in\mathcal{T}$. Hence, we define the following unknown vector:
	\begin{align}
			\boldsymbol{\eta}=&\left[\boldsymbol{\theta},\Re\{\mathbf{s}(1)\},...,\Re\{\mathbf{s}(T)\},\Im\{\mathbf{s}(1),...,\mathbf{s}(T)\}\right]\nonumber\\
			=&\left[\boldsymbol{\theta},\mathbf{s}\right]\in\mathbb{R}^{2M+2MT},
	\end{align}
	and let $\eta_i$ be the $i$-th parameter to be estimated. Then, we obtain the elements of the FIM:
	\begin{equation}\label{eq:matdec}
		\mathbf{J}_{\boldsymbol{\eta}}=\left[\begin{matrix}
			\mathbf{J}_{\boldsymbol{\theta}\boldsymbol{\theta}} & 
			\mathbf{J}_{\boldsymbol{\theta}\mathbf{s}}^{(1)}&\cdots&\mathbf{J}_{\boldsymbol{\theta}\mathbf{s}}^{(T)} \\
			\mathbf{J}_{\mathbf{s}\boldsymbol{\theta}}^{(1)} & 
			\mathbf{J}_{\mathbf{s}\mathbf{s}}^{(1)}&\mathbf{0}&\mathbf{0}\\
			\vdots &\mathbf{0}&\ddots&\mathbf{0}\\
			\mathbf{J}_{\mathbf{s}\boldsymbol{\theta}}^{(T)}&\mathbf{0}&\mathbf{0}&\mathbf{J}_{\mathbf{s}\mathbf{s}}^{(T)}
		\end{matrix}\right]\in\mathbb{R}^{2M+2MT},
	\end{equation}
	in which $\mathbf{J}_{\boldsymbol{\theta}\boldsymbol{\theta}}$ has been elaborated in Eqs. (\ref{eq:eq35f})-\eqref{eq:eq37f} and according to \cite{d2022cramer}, the $(i,j)$-th entry in $\mathbf{J}_{\mathbf{s}\mathbf{s}}^{(t)}$ and $\mathbf{J}_{\boldsymbol{\theta}\mathbf{s}}^{(t)}$ in FIM can be calculated as follows:
	\begin{align}
		\label{eq:JJIIJJ}
			&\left[\mathbf{J}_{\mathbf{s}\mathbf{s}}^{(t)}\right]_{i,j}=\frac{2}{\sigma_\nu^2}\int_{\mathcal{S}}\Re\left\{\frac{\partial H^\ast(\mathbf{r},\boldsymbol{\theta};t)}{\partial\left[\mathbf{s}(t)\right]_i}\frac{\partial H(\mathbf{r},\boldsymbol{\theta};t)}{\partial\left[\mathbf{s}(t)\right]_j}\right\}d\mathbf{r},\\
			&\label{eq:JJIIJJ2}
			\left[\mathbf{J}_{\boldsymbol{\theta}\mathbf{s}}^{(t)}\right]_{i,j}=\frac{2}{\sigma_\nu^2}\int_{\mathcal{S}}\Re\left\{\frac{\partial H^\ast(\mathbf{r},\boldsymbol{\theta};t)}{\partial\left[\mathbf{s}(t)\right]_i}\frac{\partial H(\mathbf{r},\boldsymbol{\theta};t)}{\partial\left[\boldsymbol{\theta}\right]_j}\right\}d\mathbf{r},
	\end{align}
	where $H(\mathbf{r},\boldsymbol{\theta};t)=\mathbf{a}^{\mathrm{T}}(\mathbf{r};\boldsymbol{\alpha},\boldsymbol{\phi})\mathbf{s}(t)$, $\left[\mathbf{J}_{\mathbf{s}\boldsymbol{\theta}}^{(t)}\right]_{i,j}=\left[\mathbf{J}_{\boldsymbol{\theta}\mathbf{s}}^{(t)}\right]_{j,i}$, and $i,j\in\mathcal{Q}$. 	In Eqs. (\ref{eq:JJIIJJ}) and (\ref{eq:JJIIJJ2}), the elements are specified by
	\begin{align}
			&\frac{\partial H(\mathbf{r},\boldsymbol{\theta};t)}{\partial\Re\left\{s_i(t)\right\}}=e^{\mathrm{j} k \left(r_x\sin\alpha_i\cos\phi_i+r_y\sin\phi_i\right)}, i\in \mathcal{M}\\
			&\frac{\partial H(\mathbf{r},\boldsymbol{\theta};t)}{\partial\Im\left\{s_i(t)\right\}}=\mathrm{j}e^{\mathrm{j} k \left(r_x\sin\alpha_i\cos\phi_i+r_y\sin\phi_i\right)},i\in \mathcal{M} \\
			&\frac{\partial H(\mathbf{r},\boldsymbol{\theta};t)}{\partial\alpha_i}=\mathrm{j} k \left(r_x\cos\alpha_i\cos\phi_i\right)s_i(t)\nonumber\\
			&\quad\quad\quad\quad\quad\quad \times e^{\mathrm{j} k \left(r_x\sin\alpha_i\cos\phi_i+r_y\sin\phi_i\right)},i\in \mathcal{M}\\
			&\frac{\partial H(\mathbf{r},\boldsymbol{\theta};t)}{\partial\phi_i}=\mathrm{j} k \left(r_x\sin\alpha_i\sin\phi_i-r_y\cos\phi_i\right)s_i(t)\nonumber\\
			&\quad\quad\quad\quad\quad\quad \times e^{\mathrm{j} k \left(r_x\sin\alpha_i\cos\phi_i+r_y\sin\phi_i\right)}, i\in \mathcal{M}.
	\end{align}

	According to the well known formulas on the inverse of partitioned \cite{kay1993fundamentals}, the inverse of matrix $\left[\begin{matrix}
		\mathbf{A}&\mathbf{B}\\
		\mathbf{C}&\mathbf{D}
	\end{matrix}\right]$ can be derived by Eq. (\ref{Eq:fkjz}), as shown at the bottom of the next page. 
	\begin{figure*}[b]
		\hrulefill
		\begin{equation}
			\label{Eq:fkjz}
			\left[\begin{matrix}
				\mathbf{A}&\mathbf{B}\\
				\mathbf{C}&\mathbf{D}
			\end{matrix}\right]^{-1}=
			\left[\begin{array}{cc}
				\left(\mathbf{A}-\mathbf{B} \mathbf{D}^{-1} \mathbf{C}\right)^{-1} & -\left(\mathbf{A}-\mathbf{B} \mathbf{D}^{-1} \mathbf{C}\right)^{-1} \mathbf{B} D^{-1} \\
				-\mathbf{D}^{-1} \mathbf{C}\left(\mathbf{A}-\mathbf{B} \mathbf{D}^{-1} \mathbf{C}\right)^{-1} & \mathbf{D}^{-1}+\mathbf{D}^{-1} \mathbf{C}\left(\mathbf{A}-\mathbf{B} \mathbf{D}^{-1} \mathbf{C}\right)^{-1} \mathbf{B} \mathbf{D}^{-1}
			\end{array}\right].
		\end{equation}
	\end{figure*}	
	 Hence, the CRLB of DOA estimation can be formulated as
	\begin{align}\label{eq:crlbtj}
			&\mathrm{CRLB}(\boldsymbol{\theta})=\mathbf{T}^{-1},\\
			&\mathbf{T} \triangleq \mathbf{J}_{\boldsymbol{\theta} \boldsymbol{\theta}}-\sum_{t=1}^T \mathbf{J}_{\boldsymbol{\theta} \mathbf{s}}^{(t)}\left(\mathbf{J}_{\mathbf{s s}}^{(t)}\right)^{-1} \mathbf{J}_{\mathbf{s} \boldsymbol{\theta}}^{(t)}.
	\end{align}	
	Then, the CRLB for CAPA-based DOA estimation is concluded in the following theorem:
	\begin{theorem}\label{the:the2}
		\normalfont		
		Under the situation that the target snapshots $\left\{\mathbf{s}(t)\right\}_{t=1}^T$ are unknown, the CRLB for CAPA-based DOA estimation of the $m$-th target is
		\begin{align}\label{eq::crlb2}
				&\mathrm{CRLB}_u(\alpha_m) = \left[\mathbf{T}^{-1}\right]_{m,m},m\in\mathcal{M}, \\
				&\mathrm{CRLB}_u(\phi_m) = \left[\mathbf{T}^{-1}\right]_{M+m,M+m}, m\in\mathcal{M},
		\end{align}
		where the subscript $u$ is used to highlight that the above results refer to the case of unknown target snapshots, and $\mathbf{T}$ is defined in Eq. (\ref{eq:crlbtj}).
	\end{theorem} 
	
	By using the matrix inversion lemma, $\mathbf{T}^{-1}$ can be further derived as
	\begin{align}
		\label{eq:remt-}
			&\quad\left[\mathbf{T}^{-1}\right]_{i,i}\nonumber\\
			&=\left[\mathbf{J}_{\boldsymbol{\theta} \boldsymbol{\theta}}^{-1}\right]_{i,i}+\left[\mathbf{J}_{\boldsymbol{\theta} \boldsymbol{\theta}}^{-1}\mathbf{V}(\mathbf{Z}^{-1}+\mathbf{V}^\mathrm{T}\mathbf{J}_{\boldsymbol{\theta} \boldsymbol{\theta}}^{-1}\mathbf{V})\mathbf{V}^{\mathrm{T}}\mathbf{J}_{\boldsymbol{\theta} \boldsymbol{\theta}}^{-1}
			\right]_{i,i},
	\end{align}
	where $\mathbf{Z}=\text{diag}\left\{\mathbf{J}_{\mathbf{s}\mathbf{s}}^{(1)},\mathbf{J}_{\mathbf{s}\mathbf{s}}^{(2)},...,\mathbf{J}_{\mathbf{s}\mathbf{s}}^{(T)}\right\}$, $\text{diag}\{\cdot\}$ denotes the block diagonal matrix, and $\mathbf{V}=\left[\mathbf{J}_{\boldsymbol{\theta}\mathbf{s}}^{(1)},\mathbf{J}_{\boldsymbol{\theta}\mathbf{s}}^{(2)},...,\mathbf{J}_{\boldsymbol{\theta}\mathbf{s}}^{(T)}\right]$.

	\begin{remark}\normalfont
		By comparing the results in Eqs. (\ref{eq::crlb}) and (\ref{eq:remt-}), we can analyze the relationship between CRLBs with and without knowledge of target snapshots. Notice that $\mathbf{J}_{\boldsymbol{\theta} \boldsymbol{\theta}}-\sum_{t=1}^T \mathbf{J}_{\boldsymbol{\theta} \mathbf{s}}^{(t)}\left(\mathbf{J}_{\mathbf{s s}}^{(t)}\right)^{-1} \mathbf{J}_{\mathbf{s} \boldsymbol{\theta}}^{(t)}$ is a positive semi-definite matrix because it is the $2M\times 2M$ block in the upper left corner of matrix $\mathbf{J}_{\boldsymbol{\eta}}$, which is semi-positive definite. Hence, the second term in Eq. (\ref{eq:remt-}) is also a positive semi-definite matrix and the CRLBs with and without prior knowledge of $\mathbf{s}(t)$ satisfy that:
		\begin{align}\label{rem:wq}
				&\mathrm{CRLB}_k(\alpha_m)\leq \mathrm{CRLB}_u(\alpha_m), m\in\mathcal{M},\\ &\mathrm{CRLB}_k(\phi_m)\leq \mathrm{CRLB}_u(\phi_m), m\in\mathcal{M}.
		\end{align}		
		This result indicates that the estimation accuracy of DOA worsens without the knowledge of target snapshots \cite{9906802}.
	\end{remark}
	
	Based on the above analysis, we further derive the following observation for a special case in which the equality in Eq. (\ref{rem:wq}) holds:

	\begin{proposition}\label{prop:1}
		Consider the special case with only one target. If the CAPA plane is symmetric with respect to the origin, then the CRLBs for DOA estimation are identical, regardless of whether prior knowledge of $\mathbf{s}(t)$ is available.
	\end{proposition}

	\begin{proof}
		In this case, the coupling term $\mathbf{J}_{\mathbf{s}\mathbf{s}}^{(t)}$ between the snapshot signal and steering vector can be formulated as
		\begin{align}\label{eqlmi}
				\left[\mathbf{J}_{\boldsymbol{\theta}\mathbf{s}}^{(t)}\right]_{i,j}=&\frac{2}{\sigma_\nu^2}\int_{\mathcal{S}}\Re\left\{\left|a(\mathbf{r},\boldsymbol{\theta})\right|^2s(t)\mathbf{r}^{\mathrm{T}}\frac{\partial \mathbf{d}(\alpha,\phi)}{\partial [\boldsymbol{\theta}]_i}\right\}d\mathbf{r}\nonumber\\
				=&\frac{2s(t)}{\sigma_\nu^2}\int_{\mathcal{S}}\Re\left\{\mathbf{r}^{\mathrm{T}}\frac{\partial \mathbf{d}(\alpha,\phi)}{\partial [\boldsymbol{\theta}]_i}\right\}.
		\end{align}		
		Recalling Eq. (\ref{eq:qtqyq}) and the integral term yields that
		\begin{align}\label{eq:imrx}
				&\int_{\mathcal{S}}\Re\left\{\mathbf{r}^{\mathrm{T}}\frac{\partial \mathbf{d}(\alpha,\phi)}{\partial \alpha}\right\}\nonumber\\
				=&\int_{-\frac{L_y}{2}}^{\frac{L_y}{2}}\int_{-\frac{L_x}{2}}^{\frac{L_x}{2}}r_x\cos\alpha\cos\phi dr_xdr_y=0,\\
				\label{eq:imr2x}&\int_{\mathcal{S}}\Re\left\{\mathbf{r}^{\mathrm{T}}\frac{\partial \mathbf{d}(\alpha,\phi)}{\partial \phi}\right\}\nonumber\\
				=&\int_{-\frac{L_y}{2}}^{\frac{L_y}{2}}\int_{-\frac{L_x}{2}}^{\frac{L_x}{2}}\left(r_y\cos\phi-r_x\sin\alpha\sin\phi\right) dr_xdr_y=0,
		\end{align}				
		Substituting (\ref{eq:imrx}) and \eqref{eq:imr2x} into Eq. (\ref{eq:remt-}), it can be further concluded that:
		\begin{align}\label{eq:conc58}
				\mathrm{CRLB}_k(\alpha)= \mathrm{CRLB}_u(\alpha),\\ \label{eq:conc582}
				\mathrm{CRLB}_k(\phi)= \mathrm{CRLB}_u(\phi),
		\end{align}
		which completes the proof.
	\end{proof}

	\begin{remark}\normalfont
		Although Eqs. (\ref{eq:conc58}) and \eqref{eq:conc582} holds mathematically under symmetric CAPA configuration, it does not imply that the signal and direction parameters are completely decoupled in practice. The CRLB depends on the projection of the derivative vector onto the entire signal subspace rather than this single inner product. Therefore, even in this special case, uncertainty in the signal still affects DOA estimation performance via other statistical dependencies.
	\end{remark}

	 \subsection{Comparison With SPDAs}\label{sec::sdap}
	 In this section, we compare the performance differences between CAPA and traditional SPDA-based DOA estimation methods. We simplify the problem into the DOA estimation for single target with known target snapshots, with the azimuth and elevation angles denoted by $\alpha$ and $\phi$, respectively.
	 \subsubsection{CAPA}
	The elements of FIM can be expressed as follows:
	\begin{align}
			&J_{\alpha\alpha}=k^2R_s\left[\frac{L_x^3L_y}{12}\cos^2\alpha\cos^2\phi\right],\\
			&J_{\alpha\phi}=k^2R_s\left[-\frac{L_x^3L_y}{12}\sin\alpha\cos\alpha\sin\phi\cos\phi\right],\\
			&J_{\phi\phi}=k^2R_s\left[\frac{L_x^3L_y}{12}\sin^2\alpha\sin ^2\phi+\frac{L_xL_y^3}{12}\cos^2\phi\right].
	\end{align}	
	The determinant of FIM $\det (\mathbf{J}_{\boldsymbol{\theta}\boldsymbol{\theta}})=\frac{4}{\sigma_\nu^4}\left({J}_{\alpha\phi}^2-{J}_{\alpha\alpha}{J}_{\phi\phi}\right)$ can be formulated as
	\begin{equation}\label{eq:hls}
		\begin{aligned}
			& \det (\mathbf{J}_{\boldsymbol{\theta}\boldsymbol{\theta}})=\frac{k^4L_x^4L_y^4R_s^2\cos^2\alpha\cos^4\phi}{36\sigma_{\nu}^4}
		\end{aligned}			
	\end{equation}	
	 Subsequently, the CRLB for CAPA can be simplified into
	 \begin{align}\label{crlb_sing}
	 		&\mathrm{CRLB}_{\text{CAPA}}(\alpha) =\frac{6\sigma^2_{\nu}}{k^2R_s}\left(\frac{\sin^2\alpha\sin^2\phi}{L_xL_y^3\cos^2\alpha\cos^4\phi}\right.\nonumber\\
	 		&\qquad\qquad\qquad\qquad\qquad\left.+\frac{1}{L_x^3L_y\cos^2\alpha\cos^2\phi}\right), \\
	 		&\mathrm{CRLB}_{\text{CAPA}}(\phi) = \frac{6\sigma^2_{\nu}}{k^2R_s}\left(\frac{1}{L_xL_y^3\cos^2\phi}\right).
	 \end{align}

	\subsubsection{Spatially Discrete Arrays}
	Consider a uniform planar array and its array elements are uniformly distributed on the XOY plane along the X and Y-axes with half-wavelength spacing ($d_y=d_z=\lambda/2$). Assume that the observation space is filled with short dipoles of length $\l_r$ ($l_r\ll\lambda$), vertically oriented and placed on a square grid. There are a total of $P\times Q$ elements (short dipoles) deployed on a planar with an area of $L_y\times L_x$, with $P=\lfloor L_x/d_x \rfloor$ and $Q=\lfloor L_y/d_y\rfloor$. The centers of these elements are given by the following set of points:
	\begin{equation}
		\label{Eq:aper_spda}
		\mathcal{P}=\left\{(x,y,z)|x=0,y=k_xd_x,z=k_yd_y\right\},
	\end{equation}
	where $0\leq|k_x|\leq \lfloor L_x/d_x \rfloor$ and $0\leq|k_y|\leq \lfloor L_y/d_y \rfloor$ are both integers.	
	For the $(p,q)$-th receive antenna, the observed voltage $V_{p,q}$ can be obtained by integrating over the antenna length:
	\begin{equation}
		V_{p,q}(t)=\int_{l_r}x_z(\mathbf{p},t)dz\approx l_rx(\mathbf{p}_{p,q},t)+\nu_{p,q},
	\end{equation}
	where $\mathbf{p}_{p,q}\in\mathcal{P}$ is the center of the $(p,q)$-th entry, $\nu_{p,q}$ is the independent zero-mean Gaussian random variables, with variance $\sigma_\nu^2=l_r\sigma^2$.	
	Similarly, the FIM for conventional SPDA-based DOA estimation $\mathbf{J}_{\mathcal{M}}\in \mathbb{R}^{2\times 2}$ can be derived as follows:	
	\begin{align}\label{eq:fim12}
			&\left[\mathbf{J}_{\mathcal{M}}\right]_{i,j}= \nonumber\\
			& \frac{2}{\sigma^2_{\nu}}\sum_{t=1}^{T}\Re\left\{
			\sum_{p=1}^P\sum_{q=1}^{Q} \left(\frac{\partial a_{p,q}(\boldsymbol{\theta})}{\partial[\boldsymbol{\theta}]_i}s(t)\right)^*\left(\frac{\partial a_{p,q}(\boldsymbol{\theta})}{\partial[\boldsymbol{\theta}]_j}s(t)\right)
			\right\},
	\end{align}
	where $a_{p,q}(\boldsymbol{\theta})=e^{\mathrm{j} k \mathbf{p}_{p,q}^{\mathrm{T}} \mathbf{d}({\alpha}, {\phi})}$ is the  steering vector of the $(p,q)$-th array element.	
	Subsequently, $\mathbf{J}_{\mathcal{M}}$ can be further formulated as
	\begin{align}
			&[\mathbf{J}_\mathcal{M}]_{1,1}= \Gamma  \cos ^2 \alpha\cos ^2\phi\left(P^2-1\right),\\
			&[\mathbf{J}_\mathcal{M}]_{1,2}=[\mathbf{J}_\mathcal{M}]_{2,1}=-\Gamma \sin\alpha\cos\alpha\sin\phi\cos\phi(P^2-1),\\
			& [\mathbf{J}_\mathcal{M}]_{2,2}=\Gamma \left[\sin ^2 \alpha\sin^2\phi\left(P^2-1\right)+\cos ^2 \phi\left(Q^2-1\right)\right],
	\end{align}
	where $\Gamma=\frac{R_s\pi^2l_r^2PQ}{12\sigma^2_{\nu}}$.	
%	Then, we have 
%	\begin{equation}
%		\det (\mathbf{J}_\mathcal{M})=\Gamma^2 \sin^2\alpha\cos^2\alpha(P^2-1)(Q^2-1).
%	\end{equation}
	Then, the CRLB of SPDA-based DOA estimation is given by
	\begin{align}\label{Eq:crlbMimo}
			&\mathrm{CRLB}_{\text{SPDA}}(\alpha)=\frac{12\sigma^2_{\nu}}{R_s\pi^2l_r^2PQ}\left[ \frac{\sin^2\alpha\sin^2\phi}{\cos^2\alpha\cos^4\phi(Q^2-1)}\right.\nonumber\\
			&\qquad\qquad\qquad\qquad\qquad+\left. \frac{1}{\cos^2\alpha\cos^2\phi(P^2-1)}\right],\\
			&\mathrm{CRLB}_{\text{SPDA}}(\phi)=\frac{12\sigma^2_{\nu}}{R_s\pi^2l_r^2PQ}\left[ \frac{1}{\cos^2\phi(Q^2-1)}\right].
	\end{align}
	\begin{remark}\normalfont
		\label{rem::4}
		By using the approximations $P=2L_x/\lambda$ and $Q=2L_y/\lambda$ and comparing the CRLBs of CAPA and SPDA, we observe that CAPA reduces the CRLB by a factor of $\lambda^2/8l_r^2$ in approximation. Typically, when setting $l_r=\lambda/8\pi$, this corresponds to a CRLB reduction factor of approximately  $8\pi^2$. Hence, it can be concluded that CAPA-based DOA estimation significantly outperforms the one using traditional SPDA. 
	\end{remark}

\section{Numerical Results}\label{Sec::VNR}

	In this section, we present numerical results to evaluate the performance of the proposed CAPA-MUSIC algorithm and to analyze the behavior of CAPA-based DOA estimation systems. 
	
		\begin{table}[t]
		\renewcommand\arraystretch{1.2}
		\caption{Parameter settings.}
		\centering\label{running_time}
		\fontsize{7.5}{9.5}\selectfont  	
		\label{tab:pstt}
		\begin{tabular}{|p{0.9cm}|p{5cm}|p{1.35cm}|}
			\hline
			Notations&Definitions&Values\\ \hline \hline
			${L_x}$&The width of CAPA&$1$m\\ \hline
			$L_y$&The height of CAPA&$1$m\\ \hline
			$\eta_0$&The free-space impedance&$120\pi\Omega$ \\ \hline
			$\sigma^2$&The noise power& $10^{-3}\text{V}^2/\text{m}^2$\\ \hline
			$K$& The dimension of Legendre polynomial& $30$\\ \hline
			$\lambda$ & The wavelength & $0.1$m \\ \hline
			$T$& The number of snapshots & $ 2000$\\ \hline
		\end{tabular}
	\end{table}
	
	The simulation parameters are summarized in Table \ref{tab:pstt} and are used throughout this section unless otherwise specified. All integrals are computed using Gauss–Legendre quadrature with 30 sampling points. 	All simulations were conducted on a PC equipped with an Intel i7-13980HX 2.2GHz CPU and a 32GB RAM, and the algorithms are implemented using MATLAB 2023b.

	\subsection{Convergence and Complexity of CAPA-MUSIC}
	\begin{figure}[t]
		\centering
		\includegraphics[width=0.8\linewidth]{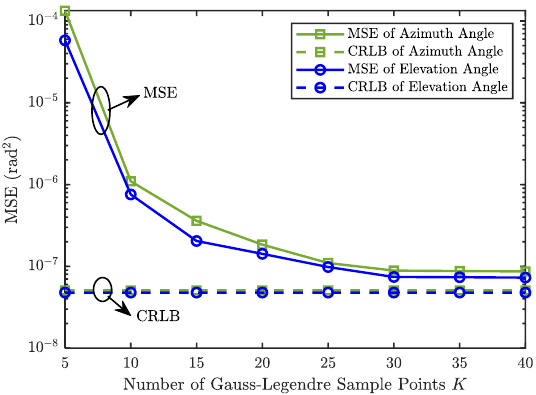}
		\caption{An illustration of the MSE and running time versus the number of Gauss-Legendre sample points $K$.}
		\label{fig:ValueKLGQ}
	\end{figure}
	
	Firstly, we evaluate the convergence performance of the proposed Gauss-Legendre numerical integration method for CAPA-MUSIC. The MSEs of DOA estimation with respect to different numbers of Gauss-Legendre sample points $K$ (i.e., the Gauss-Legendre dimension) are shown in Fig. \ref{fig:ValueKLGQ}. It can be clearly observed that as $K$ increases, the DOA estimation accuracy improves steadily. When $K$ exceeds 30, the MSEs for both azimuth and elevation angles converge to stable values that are close to the corresponding CRLB, thereby demonstrating the effectiveness of the proposed method. Therefore, selecting $K=30$ sample points is sufficient to ensure accurate numerical integration for CAPA-MUSIC, which validates the choice of $K$ in TABLE \ref{running_time}.

%	In the proposed CAPA-MUSIC algorithm, a higher Gauss-Legendre dimension contributes higher approximation accuracy for the integration in Eq. (\ref{Gauss_l}), while it inevitably involves heavier computational burden. In Fig. \ref{fig:ValueKLGQ}, we present the MSE of DOA estimation and the corresponding running time of our proposed algorithm versus $K$. The running time is averaged over 100 independent random experiments.  while also leading to longer running time. Notably, when $K\ge 20$, the MSE tends to steady while the running time sharply increases. Hence, selecting $K$ within the range $[15,20]$ can achieve a balance between estimation accuracy and computational efficiency, providing a practical guidance for choosing of $K$.

	\subsection{Performance with Single Target}

	To highlight the performance gap between the investigated CAPA and traditional SPDA systems, we present  the MSE of the proposed CRLB-MUSIC, the CRLB of  DOA estimation using CAPA, and the one using traditional SPDA in Fig. \ref{fig:msesigma}. The system model of SPDA has been elaborated in Section \ref{sec::sdap}, with the length of dipoles set as $l_r=\frac{\lambda}{4\pi}$.  The location of target is $\mathbf{z}_1=[-100,80,300]m$, the actual azimuth and elevation angles are $\alpha_1=141.34^\circ$ and $\phi_1=66.88^\circ$, respectively.

	\subsubsection{Performance Versus SNR}
	It is apparent from Fig. \ref{fig:msesigma} that the MSE and CRLB of DOA estimation gradually increase as the noise power increases. Moreover, the MSEs for both azimuth and elevation angles are close to the CRLB of CAPA-based DOA estimation, which is attributed to the proposed CAPA-MUSIC algorithm. This observation further validates the effectiveness of the proposed eigenvalue decomposition method for infinite matrices, as analyzed in \textbf{Remark \ref{rem:rem1}}. Additionally, compared to the conventional SPDA, the CAPA-based DOA estimation achieves an order-of-magnitude improvement in accuracy, which conforms to the theoretical results highlighted in \textbf{Remark \ref{rem::4}} and  validates the remarkable superiority of CAPA systems. Furthermore, since a symmetric CAPA array is employed in the simulation, the snapshot signals do not influence the CRLB of the DOA estimation (as elaborated in \textbf{Proposition \ref{prop:1}}), leading to the results in Fig. \ref{fig:msesigma}. In the remainder of the simulation, the legend entry ``CRLB of CAPA" denotes $\text{CRLB}_u$ unless otherwise specified.
	\begin{figure}[t]
		\centering
		\includegraphics[width=1\linewidth]{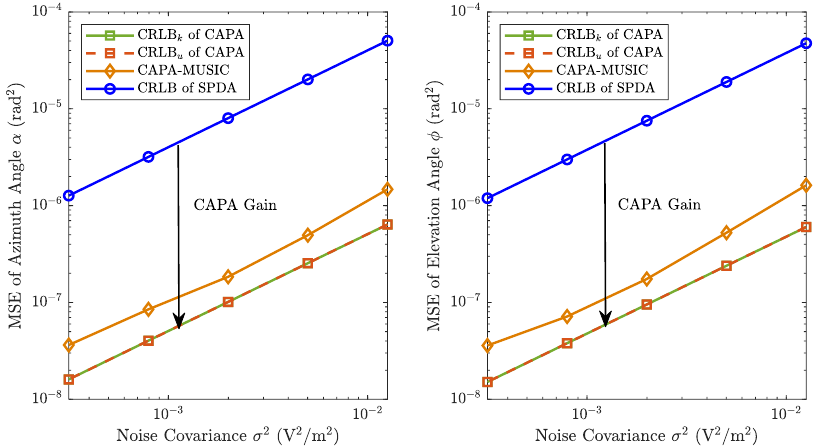}
		\caption{An illustration of the achieved MSE of CAPA-MUSIC along with CRLB.}
		\label{fig:msesigma}
	\end{figure}
	
	\subsubsection{Achieved CRLB over The DOA Plane}
	Fig. \ref{fig:crlb2d} shows the CRLB over the azimuth-elevation plane, where the achieved CRLB has been post-processed using the $\log_{10}$ function. It can be observed that:  1) When there is only one target in the scenario and the receive CAPA is square (i.e., $L_x=L_y$), the CRLB of elevation angle is independent of the target's azimuth. 2) The CRLB for azimuth estimation increases sharply as the elevation angle approaches $\pi/2$, because there is no phase difference among the receive elements across the CAPA when the target vector is perpendicular to the CAPA plane. 3) When the elevation angle approaches $0$ or $\pi$, the CRLB for azimuth estimation is optimal, while that for elevation estimation increases significantly.
	\begin{figure}[t]
		\centering
		\includegraphics[width=1\linewidth]{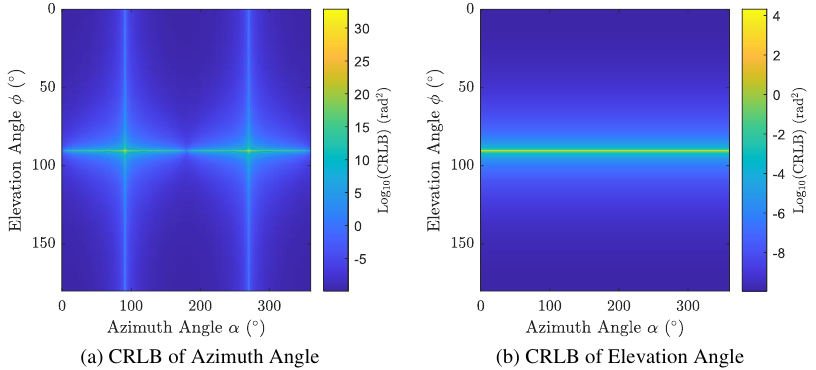}
		\caption{An illustration of achieved CRLB.}
		\label{fig:crlb2d}
	\end{figure}
	
	\subsubsection{Performance Versus Number of Snapshots}
	\begin{figure}[t]
		\centering
		\includegraphics[width=1\linewidth]{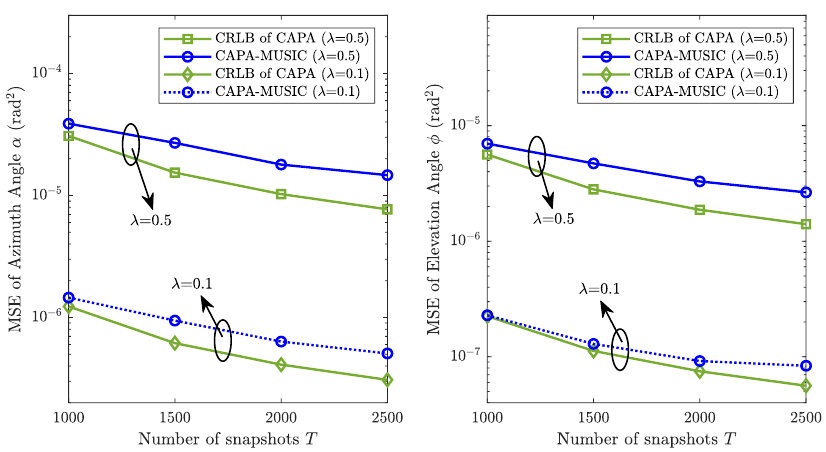}
		\caption{An illustration of achieved CRLB versus number of snapshots $T$ under different wavelength $\lambda$.}
		\label{fig:crlb_n}
	\end{figure}
	
	In Fig. \ref{fig:crlb_n}, we analyze the MSE of the proposed CAPA-MUSIC along with the corresponding CRLB as a function of number of snapshots $T$ under various wavelengths $\lambda$. From this figure, it can be observed that: 1) As the increase of $T$, both the DOA estimation MSE and the CRLB decrease, aligning well with the results in \textbf{Theorem \ref{the:the2}}. 2) As the wavelength increases, the MSE increases correspondingly, due to the fact that a shorter wavelength contributes to higher spatial resolutions. 3) Under all tested parameter settings, the proposed CAPA-MUSIC achieves performance that is very close to the CRLB for both azimuth and elevation angles, further validating its effectiveness and accuracy.

	\subsection{Performance with Multiple Targets}
	To further evaluate the CAPA-based DOA estimation performance, we present the results with multiple targets. The locations of targets are set as $\mathbf{z}_1=[50,-100,15]m$ and $\mathbf{z}_2=[200,50,15]m$.
	\subsubsection{Performance Versus SNR}
	\begin{figure}[t]
		\centering
		\includegraphics[width=1\linewidth]{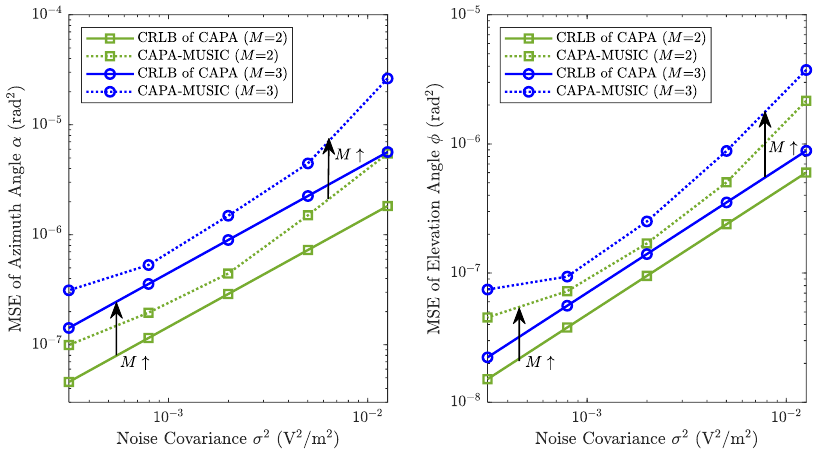}
		\caption{An illustration of achieved MSE with two targets, with the red stars representing the actual orientation of the target.}
		\label{fig:multi_targets}
	\end{figure}
	
	In Fig. \ref{fig:multi_targets}, we analyze the achieved CRLBs and MSEs as a function with regard to noise covariance $\sigma^2$. We evaluate estimation performance of both  azimuth and elevation angles under different numbers of targets $M$. For the case with $M=3$, the locations of targets are set as $\mathbf{z}_1=[50,-100,15]m$,  $\mathbf{z}_2=[200,50,15]m$, and $\mathbf{z}_3=[200,50,15]m$, respectively. It can be observed from this figure that: 1) As the number of targets to be sensed increases, both the CRLBs and MSEs of the proposed CAPA-MUSIC worsen, which is mainly attributed to the correlation between signals, especially when they are closely spaced. 2) The performance of CAPA-MUSIC is close to the corresponding CRLB in both cases, further validating the scalability of the proposed algorithm.

	\subsubsection{Performance Versus CAPA Size}
	
	\begin{figure*}[t]
		\centering
		\includegraphics[width=1\linewidth]{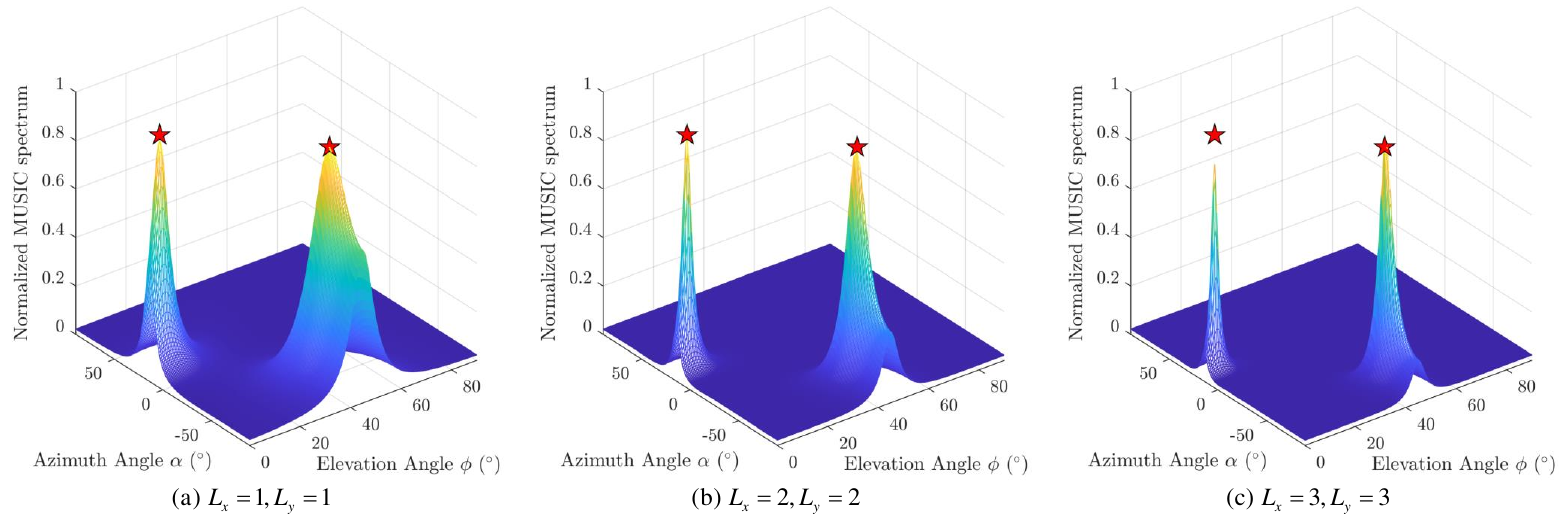}
		\caption{MUSIC spectrum for CAPA-based DOA estimation under different CAPA sizes.}
		\label{fig:spectrummusic}
	\end{figure*}
	
	The size of CAPA receiver also affects the DOA estimation accuracy. In Fig. \ref{fig:spectrummusic},  we present the MUSIC spectrum of CAPA-based DOA estimation with different CAPA sizes. Two targets are located at $[50,-100,15]$ and $[200,50,15]$, respectively. The actual azimuth and elevation angles of the targets are marked in the MUSIC spectrum with red stars. Apparently, peaks at the target point can be observed from the MUSIC spectrum under all situations. Moreover, as the CAPA size increases, the peak becomes sharper, indicating higher angle resolution and estimation accuracy. This result conforms to the theoretical results in Eq. (\ref{crlb_sing}) and further validates the effectiveness and scalability of the proposed CAPA-MUSIC algorithm.

	\section{Conclusion}\label{sec:cons}
	This paper proposed a MUSIC algorithm for CAPA-based DOA estimation and analyzed its fundamental performance limits in multiple cases. To tackle the problem of eigenvector decomposition for infinite-dimension matrices, we developed an equivalence transformation to convert the problem into the one for finite-dimension matrices. To avoid computationally intensive direct integration operations, we developed a Gauss-Legendre quadrature method to approximate the integral. Furthermore, the CRLBs for DOA estimation are derived under both known or unknown snapshot signals, validating the potential performance gain of CAPA systems compared to conventional SPDA systems. Numerical results validate the effectiveness of our proposed CAPA-MUSIC method and the potential of CAPA in DOA estimation and sensing. The designs and insights derived in this paper lay a solid foundation for future developments of CAPA-based sensing system.

	\begin{appendices}
		
	\section{Proof of Lemma \ref{lemma:1}}
\label{App:A}
	We consider the estimation of DOA parameters \(\boldsymbol{\theta} = [\alpha_1, \ldots, \alpha_M, \phi_1, \ldots, \phi_M]\) based on the received signal \(x(\mathbf{r}, t)\) over the continuous CAPA surface \(\mathcal{S}\). Define the squared residual function and its integral over the aperture as
	\begin{align}
			f(\mathbf{r}; \boldsymbol{\theta}) &= \left[ x(\mathbf{r},t) - \mathbf{a}^{\mathrm{T}}(\mathbf{r};\boldsymbol{\alpha},\boldsymbol{\phi})\mathbf{s}(t) \right]^2, \\
			F(\boldsymbol{\theta}) &= \int_{\mathcal{S}} f(\mathbf{r}; \boldsymbol{\theta})\, d\mathbf{r}.
	\end{align}	
	To prove the validity of differentiation under the integral sign (see \cite[Theorem 2.27]{folland1999real}), we verify the following conditions:
	
	\subsubsection{Parameter-Independent Domain}
	The integral domain \(\mathcal{S} \subset \mathbb{R}^2\) defined in (\ref{Eq:aper}) is fixed and compact, and does not depend on \(\boldsymbol{\theta}\), thereby satisfying the first requirement.
	
	\subsubsection{Uniform Dominated Convergence}
	The direction vector \(\mathbf{d}(\alpha_i,\phi_i)\) is continuously differentiable with respect to \(\alpha_i\) and \(\phi_i\), and its derivatives are bounded over compact subsets of the parameter domain. For any \(i \in \mathcal{M}\), the partial derivatives of \(a(\mathbf{r};\alpha_i,\phi_i)\) are computed as
	\begin{align}
			\frac{\partial a(\mathbf{r};\alpha_i,\phi_i)}{\partial\alpha_i}
			&= \mathrm{j} k \left(\mathbf{r}^{\mathrm{T}}\frac{\partial \mathbf{d}(\alpha_i,\phi_i)}{\partial\alpha_i}\right)
			e^{\mathrm{j} k\,\mathbf{r}^{\mathrm{T}}\mathbf{d}(\alpha_i,\phi_i)},\\
			\frac{\partial a(\mathbf{r};\alpha_i,\phi_i)}{\partial\phi_i}
			&= \mathrm{j} k \left(\mathbf{r}^{\mathrm{T}}\frac{\partial \mathbf{d}(\alpha_i,\phi_i)}{\partial\phi_i}\right)
			e^{\mathrm{j} k\,\mathbf{r}^{\mathrm{T}}\mathbf{d}(\alpha_i,\phi_i)}.
	\end{align}	
	Thus, there exists a constant \(C > 0\) such that
	\begin{equation}
		\left\| \frac{\partial \boldsymbol{a}(\mathbf{r}; \boldsymbol{\alpha},\boldsymbol{\phi})}{\partial [\boldsymbol{\theta}]_i} \right\|_2 \leq C(1 + \|\mathbf{r}\|_2).
	\end{equation}	
	Subsequently, we bound the derivative of \(f(\mathbf{r}; \boldsymbol{\theta})\) with respect to \([\boldsymbol{\theta}]_i\) as
	\begin{align}
			\left| \frac{\partial f(\mathbf{r}; \boldsymbol{\theta})}{\partial [\boldsymbol{\theta}]_i} \right|
			\leq 2& \left| x(\mathbf{r},t) -  \mathbf{a}^{\mathrm{T}}(\mathbf{r};\boldsymbol{\alpha},\boldsymbol{\phi}) \mathbf{s}(t) \right|\nonumber\\
			 &\left\| \frac{\partial \mathbf{a}(\mathbf{r};\boldsymbol{\alpha},\boldsymbol{\phi})}{\partial [\boldsymbol{\theta}]_i} \right\|_2 \left\|\mathbf{s}(t)\right\|_2.
	\end{align}	
	Since \(x(\mathbf{r},t)\) and \(\mathbf{s}(t)\) are bounded functions, and \(\mathcal{S}\) is compact, there exists a function \(g(\mathbf{r})=C(1+\|\mathbf{r}\|^2)\) satisfying that
	\begin{equation}
		\left| \frac{\partial f(\mathbf{r}; \boldsymbol{\theta})}{\partial [\boldsymbol{\theta}]_i} \right| \leq g(\mathbf{r}), \quad \forall \mathbf{r} \in \mathcal{S}.
	\end{equation}
	
	\subsubsection{Dominated Convergence Justification}
	Consider the difference quotient:
	\begin{equation}
		\left| \frac{f(\mathbf{r}; \boldsymbol{\theta} + \Delta\boldsymbol{\theta}) - f(\mathbf{r}; \boldsymbol{\theta})}{\Delta\boldsymbol{\theta}} \right| \leq 2 \sup_{[\boldsymbol{\theta}]_i} \left| \frac{\partial f(\mathbf{r};\boldsymbol{\theta})}{\partial [\boldsymbol{\theta}]_i} \right| \leq g(\mathbf{r}).
	\end{equation}
	We verify the integrability of \(g(\mathbf{r})\) over \(\mathcal{S}\):
	\begin{align}
			&\int_{\mathcal{S}} g(\mathbf{r}) d\mathbf{r} \nonumber\\
			=& C \int_{-L_y/2}^{L_y/2}\int_{-L_x/2}^{L_x/2} \left(r_x^2 + r_y^2\right)\, dr_x\, dr_y+CL_xL_y\nonumber \\
			=& \frac{CL_x^3 L_y + CL_x L_y^3}{12}+CL_xL_y < \infty.
	\end{align}
	Thus, \(g(\mathbf{r})\) is integrable over \(\mathcal{S}\).	By the dominated convergence theorem \cite[Theorem 2.27]{folland1999real}, we can interchange differentiation and integration, leading to
	\begin{equation}
		\frac{\partial}{\partial [\boldsymbol{\theta}]_i} \left( \int_{\mathcal{S}} f(\mathbf{r}; \boldsymbol{\theta})\, d\mathbf{r} \right) = \int_{\mathcal{S}} \frac{\partial f(\mathbf{r}; \boldsymbol{\theta})}{\partial [\boldsymbol{\theta}]_i}\, d\mathbf{r},
	\end{equation}
	which completes the proof of \textbf{Lemma \ref{lemma:1}}.
\end{appendices}

	\bibliography{literature}
\bibliographystyle{IEEEtranSHN}

\end{document}